\documentclass[a4paper]{article}

\usepackage[usenames,dvipsnames]{xcolor}
\usepackage{stmaryrd}
\usepackage{arydshln}
\usepackage{graphicx}
\usepackage{amsmath}
\usepackage{amssymb, verbatim}
\usepackage{mathrsfs}
\usepackage{dsfont}
\usepackage{url}
\usepackage[compress]{cite}
\usepackage{mdwlist}
\usepackage{paralist}
\usepackage{algorithmic}
\usepackage{algorithm}
\usepackage{tikz}
\usepackage{tkz-graph}
\usepackage[utf8]{inputenc}
\usepackage{multirow}
\usepackage{cancel}

\usetikzlibrary{arrows,positioning,automata,calc}
\usetikzlibrary{shapes,snakes}

\makeatletter
\def\url@leostyle{%
  \@ifundefined{selectfont}{\def\UrlFont{\sf}}{\def\UrlFont{\small\ttfamily}}}
\makeatother
\newtheorem{lemma}{\sc Lemma}[section]
\newtheorem{theorem}[lemma]{\sc Theorem}
\newtheorem{definition}[lemma]{\sc Definition}
\newtheorem{proposition}[lemma]{\sc Proposition}
\newtheorem{corollary}[lemma]{\sc Corollary}

\newtheorem{remark}{\sc Remark}[section]
\newtheorem{assumption}{\sc Assumption}[section]

\renewcommand{\matrix}[2]{\left[\begin{array}{#1} #2 \end{array}\right] }

\DeclareMathOperator*{\argmin}{arg\,min}

\newcommand{\G}{\mathcal{G}}
\newcommand{\Y}{\mathcal{Y}}
\newcommand{\X}{\mathcal{X}}
\newcommand{\Z}{\mathcal{Z}}
\newcommand{\V}{\mathcal{V}}
\newcommand{\E}{\mathcal{E}}
\renewcommand{\d}{\mathrm{d}}
\renewcommand{\P}{\mathcal{P}}
\newcommand{\N}{\llbracket N\rrbracket}
\newcommand{\K}{\llbracket K\rrbracket}

\newenvironment{proof}{{\noindent \bf Proof:\ }}{ \hfill $\square$}

\begin{document}

\title{A Necessary and Sufficient Condition for the Existence of Potential Functions for Heterogeneous Routing Games\footnote{An early version of this article was presented at the Annual Allerton Conference on Communication, Control, and Computing~\cite{FarokhiAllerton13}.}}

\makeatletter
\let\@fnsymbol\@arabic
\makeatother

\author{Farhad~Farokhi\thanks{F.~Farokhi and K.~H.~Johansson are with ACCESS Linnaeus Center, School of Electrical Engineering, KTH Royal Institute of Technology, SE-100 44 Stockholm, Sweden. The work of F.~Farokhi and K.~H.~Johansson was supported by grants from the Swedish Research Council, the Knut and Alice Wallenberg Foundation, and the iQFleet project. E-mails: \{farokhi,kallej\}@ee.kth.se }\,, Walid~Krichene\thanks{W.~Krichene is with the department of Electrical Engineering and Computer Sciences, University of California at Berkeley, CA, USA. E-mail: walid@eecs.berkeley.edu }$^{\hspace{.4em},}$\thanks{The work of W.~Krichene and A.~M.~Bayen was supported by grants from the California Department of Transportation, Google, and Nokia. }\,\,, Alexandre~M.~Bayen$^{4,}$\thanks{A.~M.~Bayen is with the department of Electrical Engineering and Computer Sciences, and the department of Civil and Environmental Engineering, University of California at Berkeley, CA, USA. E-mail: bayen@berkeley.edu}\,\,, and~Karl~H.~Johansson$^2$
}

\date{}

\maketitle

\begin{abstract} We study a heterogeneous routing game in which vehicles might belong to more than one type. The type determines the cost of traveling along an edge as a function of the flow of various types of vehicles over that edge. We relax the assumptions needed for the existence of a Nash equilibrium in this heterogeneous routing game. We extend the available results to present necessary and sufficient conditions for the existence of a potential function. We characterize a set of tolls that guarantee the existence of a potential function when only two types of users are participating in the game. We present an upper bound for the price of anarchy (i.e., the worst-case ratio of the social cost calculated for a Nash equilibrium over the social cost for a socially optimal flow) for the case in which only two types of players are participating in a  game with affine edge cost functions. A heterogeneous routing game with vehicle platooning incentives is used as an example throughout the article to clarify the concepts and to validate the results.
\end{abstract}

\section{Introduction} 
\subsection{Motivation}
Routing games are of special interest in transportation networks~\cite{Fisk1984301,wardrop1952road,Smith1979295} and communication networks~\cite{4278417,Altman2006286,Czumaj04selfishrouting} because they allow us to study cases in which a desirable global behavior (e.g., achieving a socially optimal solution) can emerge from simple local strategies (e.g., imposing tolls on each road based on only local information). For the purpose of this article, routing games can be decomposed into two categories. In the first category, namely, homogeneous routing games, drivers or vehicles are of the same type and, therefore, experience the same cost when using an edge in the network. Such an assumption is primarily motivated by transportation networks for which the drivers only worry about the travel time (and indeed under the assumption that all the drivers are equally sensitive to latency through considering their average behavior) or packet routing in communication networks where all the packets that are using a particular link experience the same delay or reliability. In the second category, namely, heterogeneous routing games (a.k.a., multi-class routing games~\cite{netter1971equilibrium,dafermos1972traffic}), drivers or vehicles belong to more than one type due to the following reasons:
\begin{itemize}
\item[-] \textit{Fuel Consumption}: In a transportation network, if we include the fuel consumption of the vehicles in the cost functions, two vehicles (of different types) may experience different costs for using a road even if their travel times are equal. For instance, \cite{farokhiITSC2013} studied this phenomenon in atomic congestion games in which heavy-duty vehicles experience an increased efficiency when a higher number of heavy-duty vehicles are present on the same road, because of a higher possibility of platooning and, therefore, a higher fuel efficiency, while such an increased efficiency may not be true for cars. For an experimental study of improvements in the fuel efficiency caused by platooning in heavy-duty vehicles, see~\cite{alam2010experimental}.
\item[-] \textit{Sensitivity to Latency}: It is known that drivers on a road have different sensitivities to the latency under different circumstances as well as depending on their personality and background~\cite{Stern2005,Stern199975}. In addition, due to economic advantages, heavy-duty vehicles might be more sensitive to latency in comparison to cars (because they need to deliver their goods at specific times). 
\item[-] \textit{Sensitivity to Tolls}: Drivers generally react differently to road tolls, e.g., based on the reason of the trip or their socioeconomic background. For instance, in 2001, by the request of the Swedish Institute for Transport and Communications Analysis, the consulting firm Inregia compiled a survey to estimate the value of time for the drivers in Stockholm~\cite{Inregiareport}. According to that study, drivers valued their time as 0.98, 3.30, and 0.19 SEK/min for work and school commuting trips, business trips, and other trips, respectively.
\end{itemize}
These examples motivate our interest for studying heterogeneous routing games in which the drivers or the vehicles might belong to more than one type and their type determines the cost of traveling along an edge as a function of the flow of all types of vehicles over that edge.

\subsection{Related Work}  \vspace*{-.1in}
In the context of transportation networks, routing games were originally studied in~\cite{wardrop1952road}. This study also formulated the definition of Nash equilibrium in routing games\footnote{ Throughout this article, following the convention of~\cite{roughgarden2002bad,6426543}, we use the term Nash equilibrium to refer to this equilibrium. See Remark~\ref{remark:NashWardop} for more information regarding this matter. }. Later, \cite{beckmann1956studies} showed that under some mild conditions, the routing game admits a potential function and the minimizers of this potential function are Nash equilibria of the routing game. This result guarantees the existence of a Nash equilibrium for the routing game. The problem of bounding the inefficiency of Nash equilibria has been extensively studied; see~\cite{roughgarden2007routing,roughgarden2003price,roughgarden2002bad,
braess2005paradox,awerbuch2005price,6426526} for a survey of these results.

Heterogeneous routing games have been studied extensively over the past starting with pioneering works in~\cite{netter1971equilibrium,dafermos1972traffic}. In these studies, a routing game with multi-class users were introduced and the definition of equilibrium was given. Furthermore, in~\cite{dafermos1972traffic}, the author introduced a \textit{sufficient} condition for transforming the problem of finding an equilibrium to that of an optimization (i.e., equivalent to the existence of a potential function~\cite{monderer1996potential,Rosenthal1973}). The sufficient condition is that over each edge, the increased cost of a user of the first type due to addition of one more user of the second type is equal to the increased cost of a user of the second type due to addition of one more user of the first type, i.e., the users of the first and the second type influence each other equally~\cite{dafermos1972traffic}. This condition was considered later in~\cite{EngelsonLindbergAppliedOptimization2012} in which it was also noted that satisfaction of this symmetry condition may depend on the units (e.g., time or money) adopted for representing the cost functions for the case in which the users' types are determined by their value of time (i.e., a scalar factor that balances the relationship between the latency and the imposed tolls). This results is of special interest since the equilibrium does not change by using different units for the cost functions (if the latency only depends on the sum of the flows of various types over the edge, not the individual flows, and the value of time appears linearly in the cost functions)~\cite{yang2004multi}. Necessary and sufficient conditions for the existence of potential functions in games with finite number of players was recently investigated in~\cite{deb2009testable}; however, these results were not generalized to games with a continuum of players as in heterogeneous routing games. The authors in~\cite{braess1979existence} studied the existence of an equilibrium in heterogeneous routing games even if such a symmetry condition does not hold. In contrast to these articles that assumed a finite set of types to which the users may belong, a wealth of studies also considered the case in which the users may belong to a continuum of types~\cite{leurent1993cost,Marcotte97equilibriawith}. The problem of finding tolls for general heterogeneous routing games as well as the case in which types of the users is determined by their value of time have been considered extensively~\cite{dafermos1973toll, cole2003pricing,karakostas2004edge,fleischer2004tolls,Kontogiannis2010,Marcotte2009211}. For instance, in~\cite{dafermos1973toll}, the problem of determining tolls on each edge or path for heterogeneous routing games was studied. Guarantees were provided for that the socially optimal solution (also referred to as system-optimizing flow~\cite{dafermos1972traffic}) is indeed an equilibrium of the game. However, in that article, the users were assumed to be equally sensitive to the imposed tolls. The problem of finding optimal tolls for routing game in which the users' value of time belongs to a continuum was studied in~\cite{cole2003pricing}.

\subsection{Contributions of the Article}
In this article, we formulate a general heterogeneous routing game in which the vehicles\footnote{We use the terms players, drivers, users, and vehicles interchangeably to denote an infinitesimal part of the flow that strategically tries to minimize its own cost for using the road.} might belong to more than one type. The vehicle's type determines the mapping that specifies the cost for using an edge based on the flow of all types of vehicles over that edge. 

We prove that the problem of characterizing the set of Nash equilibria for a heterogeneous routing game is equivalent to the problem of determining the set of pure strategy Nash equilibria in a game with finitely many players (in which each player represents one of the types in the original heterogeneous routing game). Doing so, we can employ classic results in game theory and economics literature, specially regarding the existence of an equilibrium in games and abstract economies~\cite{arrow1954existence,Debreu1952} (which is an extension of games to a situation where the actions of other players can modify the set of feasible actions for a player), to show that under mild conditions, a heterogeneous routing game admits at least one Nash equilibrium. 

Then, we present a necessary condition for the existence of a potential function for the heterogeneous routing game. We show that this condition is also sufficient for the case in which only two types of players are participating in the routing game. In this case, we show, following the potential game literature~\cite{monderer1996potential}, that the problem of finding a Nash equilibrium in the heterogeneous routing game can be posed as an optimization problem (which is numerically tractable if the potential function is convex). Motivated by the sufficient condition, in the rest of this article, we focus on heterogeneous routing games in which only two types of users are participating. Note that in contrast to the results of~\cite{dafermos1972traffic,EngelsonLindbergAppliedOptimization2012}, here, we present a \textit{necessary and sufficient} condition for the existence of potential function through which minimization we can recover an equilibrium. However, the price of providing this tighter condition is that we can only treat routing games with two distinct types contrary to the sufficient condition in~\cite{dafermos1972traffic,EngelsonLindbergAppliedOptimization2012}. 

If the problem of finding a Nash equilibrium in the heterogeneous routing game is numerically intractable\footnote{In general, the problem of finding a pure strategy Nash equilibrium is not numerically tractable; e.g.,~\cite{Fabrikant:2004:CPN:1007352.1007445,Papadimitriou2007agt,roughgarden2010computing}.}, it might be unlikely for the drivers to figure out a Nash equilibrium in finite time (let alone an efficient one) and utilize it. This might result in inefficient utilization of the transportation network resources. Therefore, we  present a set of tolling policies for distinguishable types (i.e., a routing game in which we may impose different tolls for different user types) and indistinguishable types (i.e., when we cannot impose type-dependent tolls) to guarantee the existence of a potential function for heterogeneous routing games. The idea of proposing tolls for indistinguishable types has been previously studied in~\cite{guo2009user}. However, in that study, the tolls were introduced to minimize the total travel time and the total travel cost (as a bi-objective optimization problem). In addition, in~\cite{guo2009user}, the users' types corresponded to socio-economic characteristics and, therefore, the cost functions of various types of users were the latency (which the function of the total flow and not individual flows of each type) plus the  tolls multiplied by the value of time.

Finally, because a Nash equilibrium is typically inefficient (i.e., it does not minimize the social cost function\footnote{We use a utilitarian social cost function (i.e., summation of the individual cost functions of all the players) as opposed to a Rawlsian social cost function (i.e., the worst-case cost function of the players); see~\cite[p.\,413]{2006general} for more information regarding the difference between these two categories of social cost functions. We present the definition of the social cost function in Section~\ref{sec:inefficiency}.}), we study the price of anarchy\footnote{The notion of price of anarchy was first introduced in~\cite{koutsoupias1999worst,papadimitriou2001algorithms}. Later, it was utilized in various games including routing games~\cite{Basar2011,johari2004efficiency,andelman2009strong,roughgarden2007routing,roughgarden2002bad}.} (a measure of the inefficiency of a Nash equilibrium which can be defined as the worst-case ratio of the social cost at a Nash equilibrium over the social cost at a socially optimal flow). We prove that for the case in which a convex potential function exists, the price of anarchy is bounded from above by two for affine edge cost functions, that is, the social cost of a Nash equilibrium can be at most twice as much as the cost of a socially optimal solution.

\subsection{Article Organization}
The rest of the article is organized as follows. We formulate the heterogeneous routing game in Section~\ref{sec:problemformulation}. In Section~\ref{sec:existence}, we prove that a Nash equilibrium may indeed exist in this routing game. We present a set of necessary and sufficient conditions to guarantee the existence of a potential function in Section~\ref{sec:findingaNashEq}. In Section~\ref{sec:tolls}, a set of tolling policies is presented to satisfy the aforementioned conditions. We bound the price of anarchy for affine cost functions in Section~\ref{sec:inefficiency}. A numerical example motivated by a heterogeneous routing game with platooning incentives is studied in Section~\ref{sec:NumericalExample}. Finally, we conclude the article and present directions for future research in Section~\ref{sec:conclusions}.

\section{A Heterogeneous Routing Game} \label{sec:problemformulation}
\subsection{Notation}
Let $\mathbb{R}$ and $\mathbb{Z}$ denote the sets of real and integer numbers, respectively. Furthermore, define $\mathbb{Z}_{\geq(\leq) a}=\{n\in\mathbb{Z}\,|\,n\geq(\leq) a\}$ and $\mathbb{R}_{\geq(\leq) a}=\{x\in\mathbb{R}\,|\,x\geq(\leq) a\}$. For simplicity of presentation, let $\mathbb{N}=\mathbb{Z}_{\geq 1}$. We use the notation $\N$ to denote $\{1,\dots,N\}$. All other sets are denoted by calligraphic letters. Specifically, $\mathcal{C}^k$ consists of all $k$-times continuously differentiable functions. 

Let $\mathcal{X}\subseteq\mathbb{R}^n$ be a set such that $0\in\mathcal{X}$. A mapping $f:\mathcal{X}\rightarrow \mathbb{R}$ is called positive definite if $f(x)\geq 0$ for all $x\in\mathcal{X}$. 

A set-valued mapping $f:\X\rightrightarrows\Y$ is said to be continuous at $x^0\in\X$ if for every  $y^0\in f(x^0)$ and every sequence $\{x^k\}_{k\in\mathbb{N}}$ such that $\lim_{k\rightarrow\infty} x^k=x^0$, there exists a sequence $\{y^k\}_{k\in\mathbb{N}}$ such that $y^k\in f(x^k)$ for all $k\in\mathbb{N}$ and $\lim_{k\rightarrow\infty} y^k=y^0$. 

We use the notation $\G=(\V,\E)$ to denote a directed graph with vertex set $\V$ and edge set $\E\subseteq\V\times\V$. Each entry $(i,j)\in\E$ denotes an edge from vertex $i\in\V$ to vertex~$j\in\V$. A directed path of length $z$ from vertex $i$ to vertex $j$ is a set of edges $\{(i_0,i_1),(i_1,i_2),\dots,(i_{z-1},i_z)\}\subseteq\E$ such that $i_0=i$ and $i_z=j$.

\subsection{Problem Formulation}
We propose an extension of the routing game introduced in~\cite{wardrop1952road} to admit more than one type of players. To be specific, we assume that the type of a player $\theta$ belongs to a finite set~$\Theta$.

Let us assume that a directed graph $\G=(\V,\E)$ models the transportation network and that a set of source--destination pairs $\{(s_k,t_k)\}_{k=0}^K$ for some constant $K\in\mathbb{N}$ are given. Each pair $(s_k,t_k)$ is called a commodity. We use the notation $\P_k$ to denote the set of all admissible paths over the graph~$\G$ that connect vertex $s_k\in\V$ (i.e., the source of this commodity) to vertex $t_k\in\V$ (i.e., the destination of this commodity). Let $\P=\cup_{k=1}^K \P_k$. We assume that each commodity $k\in\K$ needs to transfer a flow equal to $(F_k^\theta)_{\theta\in\Theta}\in\mathbb{R}_{\geq 0}^{|\Theta|}$.

We use the notation $f_p^\theta\in\mathbb{R}_{\geq 0}$ to denote the flow of players of type $\theta\in\Theta$ that use a given path $p\in\P$. We use the notation $f=(f_p^\theta)_{p\in\P,\theta\in\Theta}\in\mathbb{R}^{|\P|\cdot|\Theta|}$ to denote the aggregate vector of flows\footnote{Note that there is a one-to-one correspondence between the elements of $\P\times \Theta$ and the set of integers $\{1,\dots,|\P|\cdot|\Theta|\}$.}. A flow vector $f\in\mathbb{R}^{|\P|\cdot|\Theta|}$ is feasible if $\sum_{p\in\P_k}f_p^\theta=F_k^\theta$ for all $k\in\K$ and $\theta\in\Theta$. We use the notation $\mathcal{F}$ to denote the set of all feasible flows. To ensure that the set of feasible flows is not an empty set, we assume that $\P_k\neq \emptyset$ if $F_k^\theta\neq 0$ for any $\theta\in\Theta$. Notice that the constraints associated with each type are independent of the rest. Therefore, the flows of a specific type can be changed without breaking the feasibility of the flows associated with the rest of the types.

A vehicle of type $\theta\in\Theta$ that travels along an edge $e\in\E$ experiences a cost equal to $\tilde{\ell}_e^\theta((\phi_e^{\theta'} )_{\theta'\in\Theta})$, where for any $\theta\in\Theta$, $\phi_e^\theta$ denotes the total flow of drivers of type $\theta$ that are using this specific edge, i.e., $\phi_e^\theta=\sum_{p\in\P: e\in p}f_p^\theta$. This cost can encompass aggregates of the latency, fuel consumption, etc. For notational convenience, we assume that we can change the order with which the edge flows $\phi_e^{\theta'}$ appear as arguments of the cost function $\tilde{\ell}_e^\theta((\phi_e^{\theta'} )_{\theta'\in\Theta})$. A driver of type $\theta\in\Theta$ from commodity $k\in\K$ that uses path $p\in\P_k$ (connecting $s_k$ to $t_k$) experiences a total cost of $\ell_p^\theta(f)=\sum_{e\in p} \tilde{\ell}_e^\theta((\phi_e^{\theta'})_{\theta'\in\Theta})$. 

Each player is an infinitesimal part of the flow that tries to minimize its own cost (i.e., each player is inclined to choose the path that has the least cost). Now, based on this model, we can define the Nash equilibrium.

\begin{definition}\textsc{(Nash Equilibrium in Heterogeneous Routing Games)} \label{def:Nashrouting} A flow vector $f=(f_{p'}^{\theta'})_{p'\in\P,\theta'\in\Theta}$ is a Nash equilibrium if for all $k\in\K$ and $\theta\in\Theta$, $f_p^\theta>0$ for a path $p\in\P_k$ implies that $\ell_p^\theta(f)\leq \ell_{p'}^\theta(f)$ for all $p'\in\P_k$.
\end{definition}

This definition implies that for a commodity $k\in\K$ and type $\theta\in\Theta$, all paths with a nonzero flow for vehicles of type~$\theta$ have equal costs and the rest (i.e., paths with a zero flow for vehicles of type $\theta$) have larger or equal costs.

\begin{remark} \label{remark:NashWardop} Note that various articles use different names for the equilibrium such as, user-optim-izing flow~\cite{dafermos1972traffic,braess1979existence}, Wardrop equilibrium~\cite{Smith1979295,haurie1985relationship,braess1979existence}, Wardrop first principle~\cite{Smith1979295}, and Nash equilibrium~\cite{roughgarden2002bad,6426543}. The term Wardrop equilibrium is common, specially in transportation literature, due to the pioneering work of~\cite{wardrop1952road} as well as the fact that the term pure strategy Nash equilibrium is primarily utilized in the context of games with finitely many players~\cite{haurie1985relationship}. It is vital to note that the definition of Nash equilibrium in this paper is indeed different from that of~\cite{haurie1985relationship}, which shows that by increasing the number of users (in a game with finitely many players), the Nash equilibrium converges to the Wardrop equilibrium under appropriate assumptions. Throughout this article, following the convention of~\cite{roughgarden2002bad,6426543}, we use the term Nash equilibrium to refer to this equilibrium.  
\end{remark}

We make the following standing assumption regarding the edge latency functions for all the types.

\begin{assumption} \label{assum:1} For all $\theta\in\Theta$ and $e\in\E$, the edge cost function $\tilde{\ell}_e^\theta$ satisfies the following properties:
\begin{itemize}
\item[(\textit{i})\;\;] $\tilde{\ell}_e^{\theta}\in\mathcal{C}^1$;
\item[(\textit{ii})\;] $\tilde{\ell}_e^\theta$ is positive definite;
\item[(\textit{iii})] $\int_{0}^{\phi_e^{\theta}} \tilde{\ell}_e^{\theta}(u,(\phi_e^{\theta'})_{\theta'\in\Theta\setminus\{\theta\}}) \d u$ is a convex function in $\phi_e^{\theta}$ for any given $(\phi_e^{\theta'})_{\theta'\in\Theta\setminus\{\theta\}}$.
\end{itemize}
\end{assumption}

Assumption~\ref{assum:1}~(\textit{iii}) is equivalent to\footnote{Consult~\cite{roughgarden2007routing} for the proof of the equivalence when $|\Theta|=1$. The proof in the heterogeneous case follows the same line of reasoning.}:
\begin{itemize}
\it 
\item[(\textit{iii})'] $\tilde{\ell}_e^{\theta}(\phi_e^{\theta},(\phi_e^{\theta'})_{\theta'\in\Theta\setminus\{\theta\}}) $ is an increasing function of $\phi_e^{\theta}$ for any given $(\phi_e^{\theta'})_{\theta'\in\Theta\setminus\{\theta\}}$.
\end{itemize}
We start by proving the existence of a Nash equilibrium and, then, study the computational complexity of finding such an equilibrium. However, before that, we present an example of a heterogeneous routing game in the next subsection.

\subsection{Example: Routing Game with Platooning Incentives}  \label{subsec:Platoon}
Let $\Theta=\{\mathrm{c},\mathrm{t}\}$, where $\mathrm{t}$ denotes trucks (or, equivalently, heavy-duty vehicles) and $\mathrm{c}$ denotes cars (or, equivalently, light vehicles). Let the edge cost functions be characterized as
\begin{equation*}
\begin{split}
\tilde{\ell}_e^{\mathrm{c}}(\phi_e^\mathrm{c},\phi_e^\mathrm{t})&=\xi_e(\phi_e^\mathrm{c}+\phi_e^\mathrm{t}),\\ 
\tilde{\ell}_e^{\mathrm{t}}(\phi_e^\mathrm{c},\phi_e^\mathrm{t})&=\xi_e(\phi_e^\mathrm{c}+\phi_e^\mathrm{t})+ \zeta_e( \phi_e^\mathrm{c},\phi_e^\mathrm{t}),
\end{split}
\end{equation*}
where mappings $\xi_e:\mathbb{R}_{\geq 0}\rightarrow \mathbb{R}_{\geq 0 }$ and $\zeta_e:\mathbb{R}_{\geq 0}\times \mathbb{R}_{\geq 0}\rightarrow \mathbb{R}_{\geq 0 }$ denote respectively the latency for using edge $e$ as a function of the total flow of vehicles over that edge and the fuel consumption of trucks as a function of the flow of each type. These costs imply that cars only observe the latency $\xi_e(\phi_e^\mathrm{c}+\phi_e^\mathrm{t})$ when using the roads (which is only a function of the total flow over that edge and not the individual flows of each type). However, the cost associated with trucks encompasses an additional term which models their fuel consumption. Following this interpretation, $\zeta_e( \phi_e^\mathrm{c},\phi_e^\mathrm{t})$ is a decreasing function in $\phi_e^\mathrm{t}$ since by having a higher flow of trucks over a given road (i.e., larger $\phi_e^\mathrm{t}$) each truck gets a higher probability for collaboration such as platooning (and as a result, a higher chance of decreasing its fuel consumption). 

Let us give examples of these functions. Based on the traffic data measurements available from~\cite[p.\,366]{wardrop1952road} (see~\cite{farokhiITSC2013} for a case study on the relationship between the average velocity and the number of the vehicles on the road in Stockholm), we know that whenever the traffic on a road is in free-flow mode, we can model the average velocity of traveling along that road as an affine function of the flow of vehicles over that edge according to
$$
\bar{v}_e(\phi_e^\mathrm{c},\phi_e^\mathrm{t})=a_e(\phi_e^\mathrm{c}+\phi_e^\mathrm{t})+b_e.
$$
In this model, $b_e\in\mathbb{R}_{\geq 0}$ and $a_e\in\mathbb{R}_{\leq 0}$ for $e\in\E$. Therefore, if the length of edge $e\in\E$ is equal to $L_e\in\mathbb{R}_{\geq 0}$, we can calculate the latency of using that edge as
\begin{align*}
\xi_e(\phi_e^\mathrm{c}+\phi_e^\mathrm{t})&=\frac{L_e}{\bar{v}_e(\phi_e^\mathrm{c},\phi_e^\mathrm{t})}=\frac{L_e}{a_e(\phi_e^\mathrm{c}+\phi_e^\mathrm{t})+b_e}.
\end{align*}
Now, in cases where $a_e(\phi_e^\mathrm{c}+\phi_e^\mathrm{t})\ll b_e$, we can use a linearized\footnote{Notice that such a linearization is certainly not valid for a wide range of traffic flows, however, it models the latency functions well-enough for small flows. The authors in \cite{irwin1961capacity,irwin1962capacity} proposed a piecewise linear mapping (based on numerical data from the Toronto metropolitan area) for modeling the latency as a function of the flow of vehicles. This model justifies using a linear model for small flows (i.e., at the beginning what they call the feasible region), however, it also points out that a linear approximation is not valid for large flows. For a comprehensive comparison of different latency mappings (linear as well as nonlinear ones), see~\cite{Branston1976223}.  } model for the latency
\begin{align*}
\xi_e(\phi_e^\mathrm{c}+\phi_e^\mathrm{t})&=\frac{L_e}{b_e}-\frac{L_e a_e}{b_e^2}(\phi_e^\mathrm{c}+\phi_e^\mathrm{t}).
\end{align*}
In addition, using~\cite{Kuo-Yun2013}, we know that the total fuel consumption of a truck which is traveling with  velocity $\bar{v}_e$ for distance $L_e$ over a flat road can be modeled by
\begin{align} \label{eqn:def:zeta:nonlinear}
\zeta_e( \phi_e^\mathrm{c},\phi_e^\mathrm{t})=\frac{c_0L_e}{\bar{\eta}_{\mathrm{eng}}\rho_d} \left(\frac{1}{2}\rho_aA_ac_D\bar{v}_e^2(\phi_e^\mathrm{c},\phi_e^\mathrm{t})+mgc_r\right),
\end{align}
where $\bar{\eta}_{\mathrm{eng}}$ is the engine efficiency, $\rho_d$ is the energy density of diesel fuel, $c_D$ is the air drag coefficient, $A_a$ is the frontal area of the truck, $\rho_a$ is the air density, $m$ is the mass of the truck, $g$ is the gravitational acceleration, and $c_r$ is the the roll resistance coefficient. In addition, we have multiplied the fuel consumption by $c_0$ to balance the trade-off between the latency and fuel consumption in the aggregate cost function of the trucks. Following\cite{alam2010experimental}, we know that the air drag coefficient $c_D$ decreases if the trucks are platooning (e.g., two identical trucks can achieve 4.7\%--7.7\% reduction in the fuel consumption caused by the air drag reduction when platooning at 70\,km/h depending on the distance between them). Let us model these changes as $c_D=c_D'\gamma(\phi_e^\mathrm{t})$ where $\gamma:\mathbb{R}_{\geq 0}\rightarrow [0,1]$ is the probability of forming platoons (which is a function of the flow of trucks $\phi_e^\mathrm{t}$) multiplied by the improvements in the air drag coefficient upon platooning. Let us define parameters
$$
\alpha=\frac{L_e\rho_aA_ac'_D}{2\bar{\eta}_{\mathrm{eng}}\rho_d},\;\;\;\; \beta=\frac{L_e mgc_r}{\bar{\eta}_{\mathrm{eng}}\rho_d}.
$$
Now, again if we linearize~\eqref{eqn:def:zeta:nonlinear} around $\phi_e^{\mathrm{t}}=0$, we get
\begin{align*}
\zeta_e( \phi_e^\mathrm{c},\phi_e^\mathrm{t})=&\left( c_0\alpha\frac{\d}{\d u}\gamma(u)|_{u=0}b_e^2+2c_0\alpha\gamma(0)b_ea_e\right)\phi_e^\mathrm{t}+\left(2c_0\alpha\gamma(0)b_ea_e\right)\phi_e^\mathrm{c}+\left(c_0\beta+c_0\alpha\gamma(0)b_e^2\right).
\end{align*}
Combing all these terms results in
\begin{equation*}
\begin{split}
\tilde{\ell}_e^{\mathrm{c}}(\phi_e^\mathrm{c},\phi_e^\mathrm{t})&=\frac{L_e}{b_e}+\left(-\frac{L_e a_e}{b_e^2}\right)\phi_e^\mathrm{c}+\left(-\frac{L_e a_e}{b_e^2}\right)\phi_e^\mathrm{t},\\
\tilde{\ell}_e^{\mathrm{t}}(\phi_e^\mathrm{c},\phi_e^\mathrm{t})&=\frac{L_e}{b_e}+c_0\beta+c_0\alpha\gamma(0)b_e^2+\hspace{-.05in}\left(-\frac{L_e a_e}{b_e^2}+2c_0\alpha\gamma(0)b_ea_e\right)\phi_e^\mathrm{c}\\[-0.5em]&\;\;\;\;+\hspace{-.05in}\left(-\frac{L_e a_e}{b_e^2}+c_0\alpha\frac{\d \gamma(0)}{\d u}b_e^2\hspace{-.02in}+\hspace{-.02in}2c_0\alpha\gamma(0)b_ea_e\hspace{-.05in}\right)\hspace{-.05in}\phi_e^\mathrm{t}.
\end{split}
\end{equation*}
Notice that Assumption~\ref{assum:1}~(\textit{i}) and~(\textit{ii}) are easily satisfied. However, Assumption~\ref{assum:1}~(\textit{iii}) is only satisfied if 
$$
-\frac{L_e a_e}{b_e^2}+c_0\alpha\frac{\d \gamma(0)}{\d u}b_e^2+2c_0\alpha\gamma(0)b_ea_e\geq 0.
$$
This is indeed true because of the observation that Assumption~\ref{assum:1}~(\textit{iii}) and~(\textit{iii})' are equivalent. 

\section{Existence of Nash Equilibrium} \label{sec:existence} 
In this section, we show that the heterogeneous routing game admits a Nash equilibrium. Before stating the result, we need to introduce some concepts from~\cite{arrow1954existence} which uses results of~\cite{Debreu1952} to prove that an abstract economy (an extension of a game) admits an equilibrium under appropriate conditions\footnote{Note that we could alternatively follow the definition and results of~\cite{Debreu1952} in a direct manner, however, in that case, we need more background material presented which might be distracting to the audience.}.

\subsection{Existence of Nash Equilibrium in Games}
Let us define an abstract economy\footnote{An abstract economy was originally defined in~\cite{arrow1954existence}. It is an extension of a game. } as follows. Let $\X_i\subseteq\mathbb{R}^n$ (for some $n\in\mathbb{N}$) denote the action set of player $i\in\N$ in an abstract economy with $N$ players. We use the notation $x_i\in\X_i$ to denote the action of player~$i$. In contrast to a game, the feasible set of actions that player~$i$ can choose from is a function of actions of other players $x_{-i}=(x_j)_{j\neq i}$. Let $\Z_i:\times_{j\neq i}\X_j\rightrightarrows \X_i$ be a set-valued mapping that determines the set of feasible actions for player $i$. The utility of player $i$ is governed by a real-valued function $U_i:\times_{j=1}^N \X_j\rightarrow \mathbb{R}$. In this setup (opposed to the one presented in~\cite{arrow1954existence}), we assume the players are seeking to minimize their utility. 

\begin{definition}\textsc{(Equilibrium of an Abstract Economy~\cite{arrow1954existence})} $x^*$ is an equilibrium point of an abstract economy if, for all $i\in\N$, $x^*_i\in\argmin_{x_i\in\Z_i(x_{-i}^*)} U_i(x_i,x_{-i}^*)$. 
\end{definition}

For any $i\in\N$, we say that $\Z_i$ has a closed graph at $x_{-i}\in\times_{j\neq i}\X_j$ if the set $\{(x_j)_{j\in\N}|x_i\in\Z_i(x_{-i})\}$ is a closed set. Now, we can state the result of~\cite{arrow1954existence} regarding the existence of such an equilibrium.

\begin{theorem}[\hspace{-.1mm}\cite{arrow1954existence}] \label{tho:abstraceconomy} If, for each $i\in\N$, $\X_i$ is a compact convex set, $U_i(x_i,x_{-i})$ is continuous on $\times_{j=1}^N\X_j$ and quasi-convex in $x_i$ for each $x_{-i}\in\times_{j\neq i}\X_j$, $\Z_i$ is a continuous set-valued mapping that has a closed graph, and $\Z_i(x_{-i})$ is a nonempty convex set for each $x_{-i}\in\times_{j\neq i}\X_j$, then the abstract economy admits an equilibrium.
\end{theorem}

Note that when $\Z_i(x_{-i})=\X_i$ for all $x_{-i}\in\times_{j\neq i}\X_j$, and all $i$, we have a game with finitely many players. Therefore, an abstract economy can be considered as a generalization of a game. 

\begin{definition}\textsc{(Pure Strategy Nash Equilibrium in Games with Finitely Many Players~\cite{gibbons1992game})} \label{def:PureNash} $x^*$ is a pure strategy Nash equilibrium if, for all $i\in\N$, $x^*_i\in\argmin_{x_i\in\X_i} U_i(x_i,x_{-i}^*)$. 
\end{definition}

Theorem~\ref{tho:abstraceconomy} results now in the following useful corollary.

\begin{corollary} \label{cor:abstraceconomy} If, for each $i\in\N$, $\X_i$ is a compact convex set and $U_i(x_i,x_{-i})$ is continuous on $\times_{j=1}^N\X_j$ and quasi-convex in $x_i$ for each $x_{-i}\in\times_{j\neq i}\X_j$, then the game admits a pure strategy Nash equilibrium.
\end{corollary}

\begin{proof} The proof follows from utilizing Theorem~\ref{tho:abstraceconomy} when, for all $i\in\N$, $\Z_i(x_{-i})=\X_i$ for all $x_{-i}\in\times_{j\neq i}\X_j$. \par
\end{proof}

With this result in hand, we can go ahead and prove the existence of a Nash equilibrium in the heterogeneous routing game. In the next subsection, we first prove that the problem of finding a Nash equilibrium for the heterogeneous routing game is equivalent to the problem of finding a pure strategy Nash equilibrium in an abstract game\footnote{ We use the term ``abstract'' to emphasize the fact that the introduced game does not have any physical intuition and it is simply a mathematical concept defined for proving the results of this paper. This expression should not be confused with that of an ``abstract economy''. } with finitely many players. Then, we use Corollary~\ref{cor:abstraceconomy} to show that this game admits a Nash equilibrium under Assumption~\ref{assum:1}.

\subsection{Existence of Nash Equilibrium in Heterogeneous Routing Games}
For the sake of simplicity of presentation and without loss of generality (since $\Theta$ is finite), we can assume that $\Theta=\{\theta_1,\dots,\theta_N\}$ where $N=|\Theta|$. Now, let us define the abstract game.

\begin{definition} \label{def:highlevelgame}
An abstract game is a game with $N$ players in which player~$i\in \N$ corresponds to type $\theta_i\in\Theta$ in the heterogeneous routing game. The action of player $i$ is $a_i=(f_{p'}^{\theta_i})_{p'\in\P}$ which belongs to the action set 
$$
\mathcal{A}_i=\bigg\{(f_{p'}^{\theta_i})_{p'\in\P}\in\mathbb{R}^{|\P|}\,\Big|\,\sum_{p'\in\P_k}f_{p'}^{\theta_i}=F_k^{\theta_i}\bigg\}.
$$
Additionally, the utility of player $i$ is defined as
\begin{equation} \label{eqn:definition:U}
\begin{split}
U_i(a_i,a_{-i})&=\sum_{e\in\E}\int_0^{\phi_e^{\theta_i}} \tilde{\ell}_e^{\theta_i}(u,(\phi_e^{\theta_j})_{\theta_j\in\Theta\setminus\{\theta_i\}}) \d u,
\end{split}
\end{equation}
where $a_{-i}$ represents the actions of the rest of the players $(a_j)_{j\in\N\setminus\{i\}}$ and $\phi_e^{\theta_i}= \sum_{p\in\P:e\in p} f_p^{\theta_i}$ is the edge flow of type $\theta_i$ for each $i\in \N$. 
\end{definition}

The following result establishes an interesting relationship between the introduced abstract game and the underlying heterogeneous routing game.

\begin{lemma} \label{lemma:equivalence} A flow vector $(f_{p'}^{\theta'})_{p'\in\P,\theta'\in\Theta}$ is a Nash equilibrium of the heterogeneous routing game if and only if $((f_{p'}^{\theta_1})_{p'\in\P},\dots,(f_{p'}^{\theta_N})_{p'\in\P})$ is a pure strategy Nash equilibrium of the abstract game.
\end{lemma}

\begin{proof} Notice that $((f_{p'}^{\theta_1})_{p'\in\P},\dots,(f_{p'}^{\theta_N})_{p'\in\P})$ being a pure strategy Nash equilibrium (see Definition~\ref{def:PureNash}) of the abstract game is equivalent to that for all $i\in\N$, $a_i=(f_{p'}^{\theta_i})_{p'\in\P}$ is the best response of player $i$ to the tuple of actions $a_{-i}=((f_{p'}^{\theta_j})_{p'\in\P})_{\theta_j\in\Theta\setminus\{\theta_i\}}$ or, equivalently,
\begin{equation*}
\begin{split}
a_i\in\argmin_{(f_{p'}^{\theta_i})_{p'\in\P}} & \;\sum_{e\in\E}\int_0^{\phi_e^{\theta_i}} \tilde{\ell}_e^{\theta_i}(u,(\phi_e^{\theta_j})_{\theta_j\in\Theta\setminus\{\theta_i\}}) \d u,\\
 \mathrm{s.t.} \hspace{.1in} & \hspace{-.1in}\,\sum_{p\in\P:e\in p}f_p^{\theta_i}=\phi_e^{\theta_i}, \;\forall e\in\E,
 \\ & \sum_{p\in\P_k} f_p^{\theta_i}=F_k^{\theta_i}, \;\forall k\in\K,
% \\ & \phi_e^{\theta_i}\in\mathbb{R}_{\geq 0}, \;\forall e\in\E,
 \\ & f_p^{\theta_i}\geq 0, \;\forall p\in\P.
\end{split}
\end{equation*}
where $\phi_e^{\theta_j}=\sum_{p\in\P:e\in p}f_p^{\theta_j}$ for all $j\in\N\setminus\{i\}$. Notice that due to Assumption~\ref{assum:1}~(\textit{iii}), this problem is indeed a convex optimization problem. Let us define the Lagrangian as 
\begin{equation*}
\begin{split}
L_i((\phi_{e'}^{\theta_i})_{e'\in\E},(f_{p'}^{\theta_i})_{p'\in\P})\hspace{-.03in}=\hspace{-.03in}& \sum_{e\in\E}\int_0^{\phi_e^{\theta_i}} \hspace{-.1in} \tilde{\ell}_e^{\theta_i}(u,(\phi_e^{\theta_j})_{\theta_j\in\Theta\setminus\{\theta_i\}}\hspace{-.02in}) \d u+\sum_{e\in E} v_e^i \left(\sum_{p\in\P:e\in p}f_p^{\theta_i}-\phi_e^{\theta_i}\right)\\&-\sum_{k=1}^K w_k^i \left(\sum_{p\in\P_k} f_p^{\theta_i}-F_k^{\theta_i} \right)-\sum_{p\in\P} \lambda_p^i f_p^{\theta_i},
\end{split}
\end{equation*}
where $(v_e^i)_{e\in\E}\in\mathbb{R}^{|\mathcal{E}|}$, $(w_k^i)_{k\in\K}\in\mathbb{R}^{K}$, and $(\lambda_p^i)_{p\in\P}\in\mathbb{R}_{\geq 0}^{|\P|}$ are Lagrange multipliers. Now, using Karush--Kuhn--Tucker theorem~\cite[p.\,244]{boyd2004convex}, optimality conditions are
\begin{equation} \label{eqn:KKT:1}
\begin{split}
\frac{\partial }{\partial \phi_e^{\theta_i}} &L_i((\phi_{e'}^{\theta_i})_{e'\in\E},(f_{p'}^{\theta_i})_{p'\in\P}) =\tilde{\ell}^{\theta_i}_e(\phi_e^{\theta_i},(\phi_e^{\theta_j})_{\theta_j\in\Theta\setminus\{\theta_i\}})-v_e^i=0,\;\;\forall e\in\E,
\end{split}
\end{equation}
and
\begin{equation} \label{eqn:KKT:2}
\begin{split}
\frac{\partial }{\partial f_p^{\theta_i}}& L_i((\phi_{e'}^{\theta_i})_{e'\in\E},(f_{p'}^{\theta_i})_{p'\in\P})=\left(\sum_{e\in p} v_e^i\right)-w_k^i-\lambda_p^i=0, \;\; \forall p\in\P_k, \forall k\in\K.
\end{split}
\end{equation}
Additionally, the complimentary slackness conditions (for inequality constraints) result in $\lambda_p^if_p^i=0$ for all $p\in\P$. Hence, for all $k$ and $p\in\P_k$, we have
\begin{align}
\ell_p^{\theta_i}((f_{p'}^{\theta'})_{p'\in\P,\theta'\in\Theta})&=\sum_{e\in p} \tilde{\ell}^{\theta_i}_e(\phi_e^{\theta_i},(\phi_e^{\theta_j})_{\theta_j\in\Theta\setminus\{\theta_i\}}) & \nonumber\\&=\sum_{e\in p} v_e^i & \mbox{by~\eqref{eqn:KKT:1}} \nonumber \\&= w_k^i+\lambda_p^i. & \mbox{by~\eqref{eqn:KKT:2}} \nonumber
\end{align}
Therefore, for any $p_1,p_2\in\P_k$, if $f_{p_1}^{\theta_i},f_{p_2}^{\theta_i}>0$, we have $\lambda_{p_1}^{\theta_i}=\lambda_{p_2}^{\theta_i}=0$ (because of the complimentary slackness conditions), which results in
\begin{equation*}
\begin{split}
\ell_{p_1}^{\theta_i}((f_{p'}^{\theta'})_{p'\in\P,\theta'\in\Theta})&=w_k^i\\&=\ell_{p_2}^{\theta_i}((f_{p'}^{\theta'})_{p'\in\P,\theta'\in\Theta}).
\end{split}
\end{equation*}
Furthermore, for any $p_3\in\P_k$ such that $f_{p_3}^{\theta_i}=0$, we get $\lambda_{p_3}^{\theta_i}\geq 0$ (because of dual feasibility, i.e., the Lagrange multipliers associated with inequality constraints must be non-negative), which results in
\begin{equation*}
\begin{split}
\ell_{p_3}^{\theta_i}((f_{p'}^{\theta'})_{p'\in\P,\theta'\in\Theta})&=w_k^i+\lambda_{p_3}^{\theta_i}\\&\geq w_k^i\\&=\ell_{p_1}^{\theta_i}((f_{p'}^{\theta'})_{p'\in\P,\theta'\in\Theta}).
\end{split}
\end{equation*}
This completes the proof.
\end{proof}

\begin{theorem} \label{tho:existence} Under Assumption~\ref{assum:1}, the heterogeneous routing game admits at least one Nash equilibrium. 
\end{theorem}

\begin{proof} Following the result of Lemma~\ref{lemma:equivalence}, proving the statement of this theorem is equivalent to showing the fact that the abstract game introduced in Definition~\ref{def:highlevelgame} admits at least one pure strategy Nash equilibrium. First, notice that for all $i\in\N$, $\mathcal{A}_i$ is a non-empty, convex, and compact subset of the Euclidean space $\mathbb{R}^{|\P|}$. Second, $U_i(a_i,a_{-i})$ is continuous in all its arguments (because it is defined as an integral of a real-valued measurable function). Finally, because of Assumption~\ref{assum:1}~(\textit{iii}), $U_i(a_i,a_{-i})$ is a convex function in $a_i$. Now, it follows from Corollary~\ref{cor:abstraceconomy} that the abstract game admits at least one pure strategy Nash equilibrium.
\end{proof}

\begin{remark} \label{remark:extension} Theorem~\ref{tho:existence} can be seen as an extension of~\cite{braess1979existence}. In that study, the authors assume that the cost functions are monotone, that is, $\tilde{\ell}_e^\theta((\phi_e^{\theta'})_{\theta'\in\Theta})\leq \tilde{\ell}_e^\theta((\bar{\phi}_e^{\theta'})_{\theta'\in\Theta})$ for all $\theta\in\Theta$ if $\phi_e^{\theta'}\leq \bar{\phi}_e^{\theta'}$ if $\theta'\in\Theta$; see~\cite[p.\,58]{braess1979existence}. This condition, in turn, implies that $\tilde{\ell}_e^\theta((\phi_e^{\theta'})_{\theta'\in\Theta})$ is a non-decreasing function of \textit{all} its arguments which is stronger than Assumption~\ref{assum:1}~(\textit{iii})'.
\end{remark}

\section{Finding a Nash Equilibrium} \label{sec:findingaNashEq}
A family of games that are relatively easy to analyze are potential games. In this section, we give conditions for when the introduced abstract game is a potential game.

\begin{definition}\textsc{(Potential Game~\cite{monderer1996potential})} The abstract game is a potential game with potential function $V:\times_{i=1}^N\mathcal{A}_i\rightarrow \mathbb{R}$ if for all $i\in\N$,
\begin{equation*}
\begin{split}
V(a_i,a_{-i})-V(\bar{a}_i,a_{-i})=U_i(a_i,a_{-i})&-U_i(\bar{a}_i,a_{-i}), \\& \forall a_i,\bar{a}_i\in\mathcal{A}_i \mbox{ and } a_{-i}\in\times_{j\in\N\setminus\{i\}} \mathcal{A}_j.
\end{split}
\end{equation*}
\end{definition}

The next lemma provides a necessary condition for the existence of a potential function in $\mathcal{C}^2$.

\begin{lemma} \label{lemma:necessity} If the abstract game admits a potential function $V\in\mathcal{C}^2$, then
$$
\sum_{e\in p_1\cap p_2} \left[ \frac{\partial}{\partial \phi_e^{\theta_i}}\tilde{\ell}_e^{\theta_j}( (\phi_e^{\theta'})_{\theta'\in\Theta} ) - \frac{\partial}{\partial \phi_e^{\theta_j}}\tilde{\ell}_e^{\theta_i}( (\phi_e^{\theta'})_{\theta'\in\Theta} ) \right]=0, %\;\; \forall i,j\in\N \mbox{ and } \forall p_1,p_2\in\P.
$$
for all $i,j\in\N$ and $p_1,p_2\in\P$.
\end{lemma}	

\begin{proof} Since $V((f_{p'}^{\theta_1})_{p'\in\P},\dots,(f_{p'}^{\theta_N})_{p'\in\P})$ is a potential function for the abstract game, it satisfies, for all $i\in\N$,
\begin{equation*}
\begin{split}
V((f_{p'}^{\theta_i}&)_{p'\in\P},((f_{p'}^{\theta_j})_{p'\in\P})_{\theta_j\in\Theta\setminus\{\theta_i\}})-V((\bar{f}_{p'}^{\theta_i})_{p'\in\P},((f_{p'}^{\theta_j})_{p'\in\P})_{\theta_j\in\Theta\setminus\{\theta_i\}})
\\ &=U_i((f_{p'}^{\theta_i})_{p'\in\P},((f_{p'}^{\theta_j})_{p'\in\P})_{\theta_j\in\Theta\setminus\{\theta_i\}})-U_i((\bar{f}_{p'}^{\theta_i})_{p'\in\P},((f_{p'}^{\theta_j})_{p'\in\P})_{\theta_j\in\Theta\setminus\{\theta_i\}}),
\end{split}
\end{equation*}
which results in the identity 
\begin{equation} \label{longequation:1}
\begin{split}
&\hspace{-.4in}\frac{\partial V((f_{p'}^{\theta'})_{p'\in\P,\theta'\in\Theta})}{\partial f_{p_1}^{\theta_i}}\\[-0.5em]&=
\lim_{\epsilon\rightarrow 0}\frac{V((f_{p'}^{\theta_i}+\epsilon\delta_{p_1p'})_{p'\in\P},((f_{p'}^{\theta_j})_{p'\in\P})_{\theta_j\in\Theta\setminus\{\theta_i\}})\hspace{-.04in}-\hspace{-.04in}V((f_{p'}^{\theta_i})_{p'\in\P},((f_{p'}^{\theta_j})_{p'\in\P})_{\theta_j\in\Theta\setminus\{\theta_i\}})}{\epsilon}\\[-0.5em]&=
\lim_{\epsilon\rightarrow 0}\frac{U_i((f_{p'}^{\theta_i}+\epsilon\delta_{p_1p'})_{p'\in\P},((f_{p'}^{\theta_j})_{p'\in\P})_{\theta_j\in\Theta\setminus\{\theta_i\}})\hspace{-.04in}-\hspace{-.04in}U_i((f_{p'}^{\theta_i})_{p'\in\P},((f_{p'}^{\theta_j})_{p'\in\P})_{\theta_j\in\Theta\setminus\{\theta_i\}})}{\epsilon}\\[-0.5em]&=
\frac{\partial U_i((f_{p'}^{\theta'})_{p'\in\P,\theta'\in\Theta})}{\partial f_{p_1}^{\theta_i}}.%\\&=
%\frac{\partial}{\partial f_{p_1}^{\theta_i}} \sum_{e\in \E}%\int_0^{\phi_e^{\theta_i}} \tilde{\ell}_e^{\theta_i}(u,%(\phi_e^{\theta_j})_{\theta_j\in\Theta\setminus\{\theta_i\}}) \d u\\&=
%\sum_{e\in p_1} \tilde{\ell}_e^{\theta_i}((\phi_e^{\theta'})_{\theta'\in\Theta}),
\end{split}
\end{equation}
in which $\delta_{ij}$ denotes the Kronecker index (or delta) defined as $\delta_{ij}=1$ if $i=j$ and $\delta_{ij}=0$ otherwise. Hence, we get
\begin{equation*}
\begin{split}
\frac{\partial V((f_{p'}^{\theta'})_{p'\in\P,\theta'\in\Theta})}{\partial f_{p_1}^{\theta_i}}&\hspace{-.03in}=\hspace{-.03in}
\frac{\partial}{\partial f_{p_1}^{\theta_i}}\hspace{-.04in} \sum_{e\in \E}\hspace{-.04in}\int_0^{\phi_e^{\theta_i}} \hspace{-.08in}\tilde{\ell}_e^{\theta_i}(\hspace{-.02in}u,\hspace{-.02in}(\phi_e^{\theta_j}\hspace{-.02in})_{\theta_j\in\Theta\setminus\{\theta_i\}}) \d u\\&\hspace{-.03in}=\hspace{-.03in}
\sum_{e\in p_1} \tilde{\ell}_e^{\theta_i}((\phi_e^{\theta'})_{\theta'\in\Theta}).
\end{split}
\end{equation*}
Now, because of Clairaut-Schwarz theorem~\cite[p.\,1067]{tan2009multivariable}, we know that the following equality must hold since $V\in\mathcal{C}^2$,
\begin{equation} \label{eqn:proof:1}
\frac{\partial^2V((f_{p'}^{\theta'})_{p'\in\P,\theta'\in\Theta})}{\partial f_{p_1}^{\theta_i} \partial f_{p_2}^{\theta_j}}=
\frac{\partial^2V((f_{p'}^{\theta'})_{p'\in\P,\theta'\in\Theta})}{\partial f_{p_2}^{\theta_j} \partial f_{p_1}^{\theta_i}}.
\end{equation}
Let us calculate
\begin{equation} \label{eqn:proof:2}
\begin{split}
\frac{\partial^2V((f_{p'}^{\theta'})_{p'\in\P,\theta'\in\Theta})}{\partial f_{p_1}^{\theta_i} \partial f_{p_2}^{\theta_j}}&=\frac{\partial}{\partial f_{p_1}^{\theta_i}}\hspace{-.05in}\left[\hspace{-.03in}\frac{\partial V((f_{p'}^{\theta'})_{p'\in\P,\theta'\in\Theta})}{\partial f_{p_2}^{\theta_j}}\right]\\&=
\frac{\partial}{\partial f_{p_1}^{\theta_i}} \hspace{-.05in}\left[\sum_{e\in p_2} \tilde{\ell}_e^{\theta_j}((\phi_e^{\theta'})_{\theta'\in\Theta})\right]\\&= \sum_{e\in p_1\cap p_2} \frac{\partial \tilde{\ell}_e^{\theta_j}((\phi_e^{\theta'})_{\theta'\in\Theta})}{\partial \phi_e^{\theta_i}},
\end{split}
\end{equation}
and, similarly, 
\begin{equation} \label{eqn:proof:3}
\begin{split}
\frac{\partial^2 V((f_{p'}^{\theta'})_{p'\in\P,\theta'\in\Theta})}{\partial f_{p_2}^{\theta_j} \partial f_{p_1}^{\theta_i}}= \sum_{e\in p_1\cap p_2} \frac{\partial\tilde{\ell}_e^{\theta_i}((\phi_e^{\theta'})_{\theta'\in\Theta})}{\partial \phi_e^{\theta_j}}.
\end{split}
\end{equation}
Substituting~\eqref{eqn:proof:1} and~\eqref{eqn:proof:2} into~\eqref{eqn:proof:3} results in
$$
\sum_{e\in p_1\cap p_2} \left[\frac{\partial}{\partial \phi_e^{\theta_j}}\tilde{\ell}_e^{\theta_i}((\phi_e^{\theta'})_{\theta'\in\Theta}) - \frac{\partial}{\partial \phi_e^{\theta_i}}\tilde{\ell}_e^{\theta_j}((\phi_e^{\theta'})_{\theta'\in\Theta}) \right]=0,
$$
for all $p_1,p_2\in\P$ and $\theta_i,\theta_j\in\Theta$.
\end{proof}

Interestingly, we can prove that this condition is also a sufficient condition for the existence of a potential function (that belongs to $\mathcal{C}^2$) whenever two types of players are participating in the heterogeneous routing game.

\begin{lemma} \label{lemma:sufficiency} Assume that $|\Theta|=2$. If
$$
\sum_{e\in p_1\cap p_2} \left[\frac{\partial}{\partial \phi_e^{\theta_1}}\tilde{\ell}_e^{\theta_2}(\phi_e^{\theta_1},\phi_e^{\theta_2}) - \frac{\partial}{\partial \phi_e^{\theta_2}}\tilde{\ell}_e^{\theta_1}(\phi_e^{\theta_1},\phi_e^{\theta_2}) \right]=0,
$$
for all $p_1,p_2\in\P$, then
\begin{equation*}
\begin{split}
V((f_{p'}^{\theta_1})_{p'\in\P},(f_{p'}^{\theta_2})_{p'\in\P})\hspace{-.03in}=\hspace{-.03in}\sum_{e\in\E}\bigg[&\int_0^{\phi_e^{\theta_1}} \hspace{-.07in}\tilde{\ell}_e^{\theta_1}(u_1,\phi_e^{\theta_2}) \d u_1\hspace{-.03in}+\hspace{-.03in}\int_0^{\phi_e^{\theta_2}}\hspace{-.07in} \tilde{\ell}_e^{\theta_2}(\phi_e^{\theta_1},u_2) \d u_2 \\&\hspace{-.03in}-\hspace{-.03in}\int_0^{\phi_e^{\theta_2}}\hspace{-.07in}\int_0^{\phi_e^{\theta_1}} \hspace{-.07in}\frac{\partial}{\partial u_2 } \hspace{-.03in}\tilde{\ell}_e^{\theta_1} (u_1,u_2) \d u_1 \d u_2 \bigg]
\end{split}
\end{equation*}
is a potential function for the abstract game.
\end{lemma}

\begin{proof} 
Notice that for all $p\in\P$, we get
\begin{small}
\begin{equation} \label{longequation:2}
\begin{split}
&\hspace{-.3in}\frac{\partial V((f_{p'}^{\theta'})_{p'\in\P,\theta'\in\Theta})}{\partial f_p^{\theta_1}} \\&=\frac{\partial}{\partial f_p^{\theta_1}}\bigg(\sum_{e\in\E}\bigg[\int_0^{\phi_e^{\theta_1}} \tilde{\ell}_e^{\theta_1}(u_1,\phi_e^{\theta_2}) \d u_1\hspace{-.03in}+\hspace{-.03in}\int_0^{\phi_e^{\theta_2}} \tilde{\ell}_e^{\theta_2}(\phi_e^{\theta_1},u_2) \d u_2 \hspace{-.03in}-\hspace{-.03in}\int_0^{\phi_e^{\theta_2}}\hspace{-.05in}\int_0^{\phi_e^{\theta_1}}\hspace{-.05in} \frac{\partial}{\partial u_2 } \tilde{\ell}_e^{\theta_1} (u_1,u_2) \d u_1 \d u_2 \bigg] \bigg)\\&=
\sum_{e\in p}\bigg[ \tilde{\ell}_e^{\theta_1}(\phi_e^{\theta_1},\phi_e^{\theta_2})+\int_0^{\phi_e^{\theta_2}} \hspace{-.07in} \frac{\partial}{\partial \phi_e^{\theta_1}} \tilde{\ell}_e^{\theta_2}(\phi_e^{\theta_1},u_2) \d u_2 - \int_0^{\phi_e^{\theta_2}} \hspace{-.07in} \frac{\partial}{\partial u_2 } \tilde{\ell}_e^{\theta_1} (\phi_e^{\theta_1},u_2) \d u_2 \bigg]
\\&= \sum_{e\in p}\tilde{\ell}_e^{\theta_1}(\phi_e^{\theta_1},\phi_e^{\theta_2})+\sum_{e\in p}\int_0^{\phi_e^{\theta_2}} \left[\frac{\partial}{\partial \phi_e^{\theta_1}} \tilde{\ell}_e^{\theta_2}(\phi_e^{\theta_1},u_2)- \frac{\partial}{\partial u_2 } \tilde{\ell}_e^{\theta_1} (\phi_e^{\theta_1},u_2) \right] \d u_2.
\end{split}
\end{equation}
\end{small}
Now, let us define
\begin{equation*}
\begin{split}
\Psi((&\phi_e^{\theta_1})_{e\in\E},(\phi_e^{\theta_2})_{e\in\E})=
\sum_{e\in p}\int_0^{\phi_e^{\theta_2}} \left[\frac{\partial}{\partial \phi_e^{\theta_1}} \tilde{\ell}_e^{\theta_2}(\phi_e^{\theta_1},u)- \frac{\partial}{\partial u } \tilde{\ell}_e^{\theta_1} (\phi_e^{\theta_1},u) \right] \d u.
\end{split}
\end{equation*}
We have
\begin{equation*}
\begin{split}
\frac{\partial \Psi((\phi_e^{\theta_1})_{e\in\E},(\phi_e^{\theta_2})_{e\in\E})}{\partial f_{\hat{p}}^{\theta_2}}=
\hspace{-.1in}\sum_{e\in p\cap \hat{p}}\hspace{-.04in}\left[\frac{\partial}{\partial \phi_e^{\theta_1}} \tilde{\ell}_e^{\theta_2}(\phi_e^{\theta_1},\phi_e^{\theta_2})\hspace{-.03in}-\hspace{-.03in} \frac{\partial}{\partial u } \tilde{\ell}_e^{\theta_1} (\phi_e^{\theta_1},\phi_e^{\theta_2}) \right]\hspace{-.03in}=\hspace{-.03in}0,
\end{split}
\end{equation*}
for all $\hat{p}\in\P$. Noticing that $\phi_e^{\theta_2}=\sum_{\hat{p}\in\P:e\in \hat{p}} f_{\hat{p}}^{\theta_2}$ for all $e\in\E$, we get
\begin{equation*}
\begin{split}
\frac{\partial \Psi((\phi_e^{\theta_1})_{e\in\E},(\phi_e^{\theta_2})_{e\in\E})}{\partial \phi_e^{\theta_2}}&= \hspace{-.1in}\sum_{\hat{p}\in\P:e\in \hat{p}} \hspace{-.1in}\frac{\partial \Psi((\phi_e^{\theta_1})_{e\in\E},(\phi_e^{\theta_2})_{e\in\E})}{\partial f_{\hat{p}}^{\theta_2}}=0, \;\;\; \forall e\in\E.
\end{split}
\end{equation*}
Thus, $\Psi((\phi_e^{\theta_1})_{e\in\E},(\phi_e^{\theta_2})_{e\in\E}) = \Psi((\phi_e^{\theta_1})_{e\in\E},0) = 0$. Setting $\Psi((\phi_e^{\theta_1})_{e\in\E},(\phi_e^{\theta_2})_{e\in\E})=0$ (see definition above) inside~\eqref{longequation:2} results in
\begin{equation} \label{eqn:equality:derivativeVandU1}
\begin{split}
\frac{\partial V((f_{p'}^{\theta'})_{p'\in\P,\theta'\in\Theta})}{\partial f_p^{\theta_1}} &= \sum_{e\in p}\tilde{\ell}_e^{\theta_1}(\phi_e^{\theta_1},\phi_e^{\theta_2})
= \frac{\partial U_1((f_{p'}^{\theta_1})_{p'\in\P},(f_{p'}^{\theta_2})_{p'\in\P})}{\partial f_p^{\theta_1}} ,
\end{split}
\end{equation}
where the partial derivatives of $U_1$ can be computed from its definition in~\eqref{eqn:definition:U}.  Let $(f_{p'}^{\theta_1})_{p'\in\P}$ and $(\bar{f}_{p'}^{\theta_1})_{p'\in\P}$ be arbitrary points in set of actions $\mathcal{A}_1$.  Furthermore, let $r:[0,1]\rightarrow \mathcal{A}_1$ be a continuously differentiable mapping (i.e., $r\in\mathcal{C}^1$) such that $r(0)=(\bar{f}_{p'}^{\theta_1})_{p'\in\P}$ and $r(1)=(f_{p'}^{\theta_1})_{p'\in\P}$ which remains inside $\mathcal{A}_1\subseteq\mathbb{R}^{|\P|}$ for all $t\in(0,1)$. We define $\mathrm{graph}(r)$ as the collection of all ordered pairs $(t,r(t))$ for all $t\in[0,1]$, which denotes a continuous path that connects $(f_{p'}^{\theta_1})_{p'\in\P}$ and $(\bar{f}_{p'}^{\theta_1})_{p'\in\P}$. We know that at least one such mapping exists because $\mathcal{A}_1$ is a simply connected set for all $i\in\N$. Hence, we have 
\begin{equation*}
\begin{split}
\int_{\mathrm{graph}(r)} \left[\frac{\partial V(a_1,a_2)}{\partial a_1}\bigg|_{a_1=r}\right]^\top \d r
&=\int_0^1 \left[\frac{\partial V(a_1,a_2)}{\partial a_1}\bigg|_{a_1=r(t)}\right]^\top\frac{\partial r(t)}{\partial t} \d t
\\[-.1em]&=\int_0^1 \left[\frac{\d}{\d t}V(r(t),a_2)\right] \d t
\\[-.1em]&=V(r(1),a_2)-V(r(0),a_2)
\\[-.1em]&=V((f_{p'}^{\theta_1})_{p'\in\P},(f_{p'}^{\theta_2})_{p'\in\P}) -V((\bar{f}_{p'}^{\theta_1})_{p'\in\P},(f_{p'}^{\theta_2})_{p'\in\P}),
\end{split}
\end{equation*}
where the second to last equality is a direct consequence of the fundamental theorem of calculus~\cite[p.\,1257]{tan2009multivariable}. Note that this equality holds irrespective of the selected path. Therefore, 
\begin{equation*}
\begin{split}
V((f_{p'}^{\theta_1})_{p'\in\P},(f_{p'}^{\theta_2})_{p'\in\P}) &-V((\bar{f}_{p'}^{\theta_1})_{p'\in\P},(f_{p'}^{\theta_2})_{p'\in\P})\\ &=\int_{\mathrm{graph}(r)} \left[\frac{\partial V(a_1,a_2)}{\partial a_1}\bigg|_{a_1=r}\right]^\top\d r\\&=\int_{\mathrm{graph}(r)} \left[\frac{\partial U_1(a_1,a_2)}{\partial a_1}\bigg|_{a_1=r}\right]^\top\d r \hspace{.8in} \mbox{by~\eqref{eqn:equality:derivativeVandU1}}\\ &=U_1((f_{p'}^{\theta_1})_{p'\in\P},(f_{p'}^{\theta_2})_{p'\in\P}) -U_1((\bar{f}_{p'}^{\theta_1})_{p'\in\P},(f_{p'}^{\theta_2})_{p'\in\P}),
\end{split}
\end{equation*}
Similarly, we can also prove
\begin{small}
\begin{equation} \label{longequation:3}
\begin{split}
&\hspace{-.2in}\frac{\partial V((f_{p'}^{\theta'})_{p'\in\P,\theta'\in\Theta})}{\partial f_p^{\theta_2}} 
\\&=\frac{\partial}{\partial f_p^{\theta_2}}\bigg(\sum_{e\in\E}\bigg[\int_0^{\phi_e^{\theta_1}} \hspace{-.05in}\tilde{\ell}_e^{\theta_1}(u_1,\phi_e^{\theta_2}) \d u_1\hspace{-.03in}+\hspace{-.03in}\int_0^{\phi_e^{\theta_2}} \hspace{-.05in}\tilde{\ell}_e^{\theta_2}(\phi_e^{\theta_1},u_2) \d u_2 \hspace{-.03in}-\hspace{-.03in}\int_0^{\phi_e^{\theta_2}}\hspace{-.07in}\int_0^{\phi_e^{\theta_1}} \hspace{-.05in}\frac{\partial}{\partial u_2 } \tilde{\ell}_e^{\theta_1} (u_1,u_2) \d u_1 \d u_2 \bigg] \bigg)\\&=
\sum_{e\in p}\bigg[\int_0^{\phi_e^{\theta_1}} \hspace{-.07in} \frac{\partial}{\partial \phi_e^{\theta_2}} \tilde{\ell}_e^{\theta_1}(u_1,\phi_e^{\theta_2}) \d u_1 + \tilde{\ell}_e^{\theta_2}(\phi_e^{\theta_1},\phi_e^{\theta_2}) - \int_0^{\phi_e^{\theta_1}} \hspace{-.07in} \frac{\partial}{\partial \phi_e^{\theta_2}} \tilde{\ell}_e^{\theta_1} (u_1,\phi_e^{\theta_2}) \d u_1 \bigg]\\&=\sum_{e\in p}\tilde{\ell}_e^{\theta_2}(\phi_e^{\theta_1},\phi_e^{\theta_2}) ,
\end{split}
\end{equation}
\end{small}
which results in
\begin{equation*}
\begin{split}
\frac{\partial V((f_{p'}^{\theta_1})_{p'\in\P},(f_{p'}^{\theta_2})_{p'\in\P})}{\partial f_p^{\theta_2}} &\hspace{-.05in}=\hspace{-.05in} \sum_{e\in p}\tilde{\ell}_e^{\theta_2}(\phi_e^{\theta_1},\phi_e^{\theta_2})
\hspace{-.05in}=\hspace{-.05in} \frac{\partial U_2((f_{p'}^{\theta_1})_{p'\in\P},(f_{p'}^{\theta_2})_{p'\in\P})}{\partial f_p^{\theta_2}},
\end{split}
\end{equation*}
and, consequently, 
\begin{equation*}
\begin{split}
V((f_{p'}^{\theta_1})_{p'\in\P},(f_{p'}^{\theta_2})_{p'\in\P}) &-V((f_{p'}^{\theta_1})_{p'\in\P},(\bar{f}_{p'}^{\theta_2})_{p'\in\P})\\&=U_2((f_{p'}^{\theta_1})_{p'\in\P},(f_{p'}^{\theta_2})_{p'\in\P})-U_2((f_{p'}^{\theta_1})_{p'\in\P},(\bar{f}_{p'}^{\theta_2})_{p'\in\P}).
\end{split}
\end{equation*}
This concludes the proof.
\end{proof}

Now, combing the previous two lemmas results in the main result of this section.

\begin{theorem} \label{tho:combined} Assume that $|\Theta|=2$. The abstract game admits a potential function $V\in\mathcal{C}^2$ if and only if
$$
\sum_{e\in p_1\cap p_2} \left[ \frac{\partial}{\partial \phi_e^{\theta_1}}\tilde{\ell}_e^{\theta_2}(\phi_e^{\theta_1},\phi_e^{\theta_2}) - \frac{\partial}{\partial \phi_e^{\theta_2}}\tilde{\ell}_e^{\theta_1}(\phi_e^{\theta_1},\phi_e^{\theta_2}) \right]=0,
$$
for all $p_1,p_2\in\P$.
\end{theorem}

\begin{proof} The proof easily follows from Lemmas~\ref{lemma:necessity} and~\ref{lemma:sufficiency}. Note that the potential function presented in Lemma~\ref{lemma:sufficiency} belongs to $\mathcal{C}^2$ due to Assumption~\ref{assum:1}~(\textit{i}).
\end{proof}

Following a basic property of potential games, it is easy to prove the following corollary which shows that the process of finding a Nash equilibrium of the heterogeneous routing game is equivalent to solving an optimization problem. 

\begin{corollary} \label{cor:optimization} Assume that $|\Theta|=2$. Furthermore, let 
$$
\sum_{e\in p} \left[ \frac{\partial}{\partial \phi_e^{\theta_1}}\tilde{\ell}_e^{\theta_2}(\phi_e^{\theta_1},\phi_e^{\theta_2}) - \frac{\partial}{\partial \phi_e^{\theta_2}}\tilde{\ell}_e^{\theta_1}(\phi_e^{\theta_1},\phi_e^{\theta_2}) \right]=0,
$$
for all $p \in\P$. If $f=(f_{p'}^{\theta'})_{p'\in\P,\theta'\in\Theta}$ is a solution of the optimization problem
\begin{equation*}
\begin{split}
\min \hspace{.1in} & V((f_{p'}^{\theta_1})_{p'\in\P},(f_{p'}^{\theta_2})_{p'\in\P}),  \\
\mathrm{s.t.} \hspace{.1in}
& \hspace{-.1in}\sum_{p\in\P:e\in p}f_p^{\theta_1}=\phi_e^{\theta_1} \mbox{ and } \hspace{-.1in}\sum_{p\in\P:e\in p}f_p^{\theta_2}=\phi_e^{\theta_2}, \;\forall e\in\E, \\
& \sum_{p\in\P_k} f_p^{\theta_1}=F_k^{\theta_1} \mbox{ and } \sum_{p\in\P_k} f_p^{\theta_2}=F_k^{\theta_2}, \;\forall k\in\K, \\
%& \phi_e^{\theta_1},\phi_e^{\theta_2}\in\mathbb{R}_{\geq 0}, \;\forall e\in\E, \\
& f_p^{\theta_1},f_p^{\theta_2}\geq 0, \;\forall p\in\P,
\end{split}
\end{equation*}
where $V((f_{p'}^{\theta_1})_{p'\in\P},(f_{p'}^{\theta_2})_{p'\in\P})$ is defined in Lemma~\ref{lemma:sufficiency}, then $f=(f_p^\theta)_{p\in\P,\theta\in\Theta}$ is a Nash equilibrium of the heterogeneous routing game.
\end{corollary}

\begin{proof} The proof is consequence of the fact that a minimizer of the potential function is a pure strategy Nash equilibrium of a potential game; see~\cite{monderer1996potential}.
\end{proof}

Notice that so far we have proved that a minimizer of the potential function is a Nash equilibrium but not the other way round. Now, we are ready to prove this whenever the potential function is convex. 

\begin{corollary} \label{cor:convexprogramming} Let $|\Theta|=2$ and
$$
\frac{\partial}{\partial \phi_e^{\theta_1}}\tilde{\ell}_e^{\theta_2}(\phi_e^{\theta_1},\phi_e^{\theta_2})=\frac{\partial}{\partial \phi_e^{\theta_2}}\tilde{\ell}_e^{\theta_1}(\phi_e^{\theta_1},\phi_e^{\theta_2}),
$$
for all $e\in\E$. Furthermore, assume that the potential function $V((f_{p'}^{\theta_1})_{p'\in\P},(f_{p'}^{\theta_2})_{p'\in\P})$, defined in Lemma~\ref{lemma:sufficiency}, is a convex function. Then $f=(f_{p'}^{\theta'})_{p'\in\P,\theta'\in\Theta}$ is a Nash equilibrium of the heterogeneous routing game if and only if it is a solution of the convex optimization problem
\begin{equation*}
\begin{split}
\min \hspace{.1in} & \;V((f_{p'}^{\theta_1})_{p'\in\P},(f_{p'}^{\theta_2})_{p'\in\P}),  \\
\mathrm{s.t.} \hspace{.1in}
& \hspace{-.1in}\sum_{p\in\P:e\in p}f_p^{\theta_1}=\phi_e^{\theta_1} \mbox{ and } \hspace{-.1in}\sum_{p\in\P:e\in p}f_p^{\theta_2}=\phi_e^{\theta_2}, \;\forall e\in\E, \\
& \sum_{p\in\P_k} f_p^{\theta_1}=F_k^{\theta_1} \mbox{ and } \sum_{p\in\P_k} f_p^{\theta_2}=F_k^{\theta_2}, \;\forall k\in\K, \\
%& \phi_e^{\theta_1},\phi_e^{\theta_2}\in\mathbb{R}_{\geq 0}, \;\forall e\in\E, \\
& f_p^{\theta_1},f_p^{\theta_2}\geq 0, \;\forall p\in\P.
\end{split}
\end{equation*}
\end{corollary}

\begin{proof} See~Appendix~\ref{appendix:A}. \end{proof}

\begin{remark} Note that Corollary~\ref{cor:convexprogramming} is proved at the price of a more conservative condition because the conditions in Corollary~\ref{cor:optimization} requires the summation of the differences between the derivatives of the cost functions to be equal to zero while Corollary~\ref{cor:convexprogramming} needs the individual differences to be equal to zero. Notice that Corollary~\ref{cor:convexprogramming} provides the same sufficient condition for characterizing the set of all equilibria as in~\cite{EngelsonLindbergAppliedOptimization2012,dafermos1972traffic}, but these references handle the general case of $|\Theta|\geq 2$ (specifically, see Proposition~1 and Theorem~1 in~\cite{EngelsonLindbergAppliedOptimization2012}). Therefore, we can see that the presented condition in Corollary~\ref{cor:optimization} is tighter than the results of~\cite{EngelsonLindbergAppliedOptimization2012,dafermos1972traffic} (since it is also a necessary condition); however, it is only valid for $|\Theta|=2$ in contrast. 
\end{remark}

\subsection{Example: Routing Game with Platooning Incentives} \label{subsec:Platoon:potential}
Let us examine the implications of Corollary~\ref{cor:convexprogramming} in the routing game with platooning incentives in Subsection~\ref{subsec:Platoon}. For the linearized model, we can easily calculate that
\begin{equation} \label{eqn:derivative:1}
\frac{\partial \tilde{\ell}_e^{\mathrm{c}}(\phi_e^{\mathrm{c}},\phi_e^{\mathrm{t}})}{\partial \phi_e^{\mathrm{t}}} = -L_ea_e/b_e^2,
\end{equation}
\begin{equation} \label{eqn:derivative:2}
\begin{split}
\frac{\partial \tilde{\ell}_e^{\mathrm{t}}(\phi_e^{\mathrm{c}},\phi_e^{\mathrm{t}})}{\partial \phi_e^{\mathrm{c}}} &= -L_ea_e/b_e^2+2c_0\alpha \gamma(0)b_ea_e\\&= -L_ea_e/b_e^2+2c_0\alpha b_ea_e.
\end{split}
\end{equation}
where the second equality follows from $\gamma_e(0)=1$, which holds because from the definition of the mapping $\gamma:\mathbb{R}_{\geq 0}\rightarrow [0,1]$, we know that in this case (i.e., when no trucks are using that edge) the air drag coefficient is equal to its nominal value. Therefore, the condition of Corollary~\ref{cor:convexprogramming} does not hold (unless $c_0=0$). Noting that if the problem of finding a Nash equilibrium in the heterogeneous routing game is numerically intractable, it might be highly unlikely for the drivers to figure out a Nash equilibrium in reasonable time (let alone an efficient one) and utilize it, which might result in wasting parts of the transportation network resources. Therefore, a natural question that comes to mind is whether it is possible to guarantee the existence of a potential function for a heterogeneous routing game by imposing appropriate tolls.

\section{Imposing Tolls to Guarantee the Existence of a Potential Function} \label{sec:tolls}
\subsection{Definition and Results}
Let us assume that a vehicle of type $\theta\in\Theta$ must pay a toll $\tilde{\tau}_e^\theta((\phi_e^{\theta'} )_{\theta'\in\Theta})$ for using an edge $e\in\E$, where $\phi_e^\theta=\sum_{p\in\P:e\in p} f_p^\theta$. Therefore, a vehicle using path $p\in\P_k$ endures a total cost of $\ell_p^\theta(f)+\tau_p^\theta(f)$, where $\tau_p(f)$ is the total amount of money that the vehicle must pay for using path $p$ and can be calculated as $\tau_p^\theta(f)=\sum_{e\in p}\tilde{\tau}_e^\theta((\phi_e^{\theta'})_{\theta'\in\Theta})$. The definition of a Nash equilibrium is slightly modified to account for the tolls. 

\begin{definition}\textsc{(Nash Equilibrium in Heterogeneous Routing Game with Tolls)} \label{def:Nashroutingtax} A flow vector $f=(f_{p'}^{\theta'})_{p'\in\P,\theta'\in\Theta}$ is a Nash equilibrium for the routing game with tolls if, for all $k\in\K$ and $\theta\in\Theta$, whenever $f_p^\theta>0$ for some path $p\in\P_k$, then $\ell_p^\theta(f)+\tau_p^\theta(f)\leq \ell_{p'}^\theta(f)+\tau_{p'}^\theta(f)$ for all $p'\in\P_k$. 
\end{definition}

Before stating the main result of this section, note that we can have both distinguishable and indistinguishable types. This characterization is of special interest when considering the implementation of tolls. For distinguishable types, we can impose individual tolls for each type. However, for indistinguishable types, the tolls are independent of the type. To give an example, if $\Theta=\{\mathrm{cars},\mathrm{trucks}\}$, we can impose different tolls for each group of vehicles while if $\Theta=\{\mathrm{patient\;drivers},\mathrm{impatient\;drivers}\}$, we cannot. Notice that in the case of indistinguishable types, one might argue that we cannot measure $\phi_e^{\theta_i}$ for each $\theta_i\in\Theta$ individually (because as we motivated the type of user may not be identified from physical traits). However, we can use surveys and historical data to extract the statistics of each type (e.g, to realize what ratio of the actual flow belongs to each type) but when calculating the tolls for each user we cannot force that user to participate in a survey. We treat these two cases separately. 

\begin{proposition} \textsc{(Distinguishable Types)} \label{prop:toll1} Assume that $|\Theta|=2$. The abstract game admits the potential function 
\begin{equation*}
\begin{split}
V((f_{p'}^{\theta_1})_{p'\in\P},(f_{p'}^{\theta_2})_{p'\in\P})=&\sum_{e\in\E}\bigg[\int_0^{\phi_e^{\theta_1}} \hspace{-.07in}(\tilde{\ell}_e^{\theta_1}(u_1,\phi_e^{\theta_2})+\tilde{\tau}_e^{\theta_1}(u_1,\phi_e^{\theta_2})) \d u_1\\&\hspace{.1in}+\hspace{-.03in}\int_0^{\phi_e^{\theta_2}}\hspace{-.07in} (\tilde{\ell}_e^{\theta_2}(\phi_e^{\theta_1},u_2)+\tilde{\tau}_e^{\theta_2}(\phi_e^{\theta_1},u_2)) \d u_2 \\&\hspace{.1in}-\hspace{-.03in}\int_0^{\phi_e^{\theta_2}}\hspace{-.07in}\int_0^{\phi_e^{\theta_1}} \hspace{-.07in}\frac{\partial}{\partial u_2 } (\tilde{\ell}_e^{\theta_1} (u_1,u_2)+\tilde{\tau}_e^{\theta_1} (u_1,u_2)) \d u_1 \d u_2 \bigg]
\end{split}
\end{equation*}
if 
\begin{equation*}
\begin{split}
\frac{\partial \tilde{\tau}_e^{\theta_1}(\phi_e^{\theta_1},\phi_e^{\theta_2})}{\partial \phi_e^{\theta_2}}-\frac{\partial \tilde{\tau}_e^{\theta_2}(\phi_e^{\theta_1},\phi_e^{\theta_2})}{\partial \phi_e^{\theta_1}}=\frac{\partial \tilde{\ell}_e^{\theta_2}(\phi_e^{\theta_1},\phi_e^{\theta_2})}{\partial \phi_e^{\theta_1}}-\frac{\partial \tilde{\ell}_e^{\theta_1}(\phi_e^{\theta_1},\phi_e^{\theta_2})}{\partial \phi_e^{\theta_2}},
\end{split}
\end{equation*}
for all $e\in\E$. 
\end{proposition}

\begin{proof} See Appendix~\ref{appendix:B}.
\end{proof}

\begin{proposition} \textsc{(Indistinguishable Types)} \label{prop:toll2} Assume that $|\Theta|=2$. The abstract game admits the potential function $V\in\mathcal{C}^2$ in~Proposition~\ref{prop:toll1} with $\tilde{\tau}_e^{\theta_1}(\phi_e^{\theta_1},\phi_e^{\theta_2})=\tilde{\tau}_e^{\theta_2}(\phi_e^{\theta_1},\phi_e^{\theta_2})=\tilde{\tau}_e(\phi_e^{\theta_1},\phi_e^{\theta_2})$ if 
\begin{equation*}
\begin{split}
\frac{\partial \tilde{\tau}_e(\phi_e^{\theta_1},\phi_e^{\theta_2})}{\partial \phi_e^{\theta_2}}-\frac{\partial \tilde{\tau}_e(\phi_e^{\theta_1},\phi_e^{\theta_2})}{\partial \phi_e^{\theta_1}}=\frac{\partial \tilde{\ell}_e^{\theta_2}(\phi_e^{\theta_1},\phi_e^{\theta_2})}{\partial \phi_e^{\theta_1}}-\frac{\partial \tilde{\ell}_e^{\theta_1}(\phi_e^{\theta_1},\phi_e^{\theta_2})}{\partial \phi_e^{\theta_2}},
\end{split}
\end{equation*}
for all $e\in\E$. 
\end{proposition}

\begin{proof} The proof immediately follows from using Proposition~\ref{prop:toll1} with the constraint that the tolls may not depend on the type, i.e., $\tilde{\tau}_e^{\theta_1}(\phi_e^{\theta_1},\phi_e^{\theta_2})=\tilde{\tau}_e^{\theta_2}(\phi_e^{\theta_1},\phi_e^{\theta_2})=\tilde{\tau}_e(\phi_e^{\theta_1},\phi_e^{\theta_2})$.
\end{proof}

In general, we can prove the following corollary concerning the type-independent tolls.

\begin{corollary} (\textsc{Indistingushable Types}) \label{cor:toll3} Assume that $|\Theta|=2$.  The abstract game admits a potential function $V\in\mathcal{C}^2$ if  the imposed tolls are of the following form
$$
\tilde{\tau}_e(\phi_e^{\theta_1},\phi_e^{\theta_2})=c_e+\int_0^{\phi_e^{\theta_2}}\hspace{-.13in} f_e(q,\phi_e^{\theta_1}+\phi_e^{\theta_2}-q) \d q+\psi_e(\phi_e^{\theta_1}+\phi_e^{\theta_2}),
$$
where $c_e\in\mathbb{R}_{\geq 0}$, $\psi_e\in\mathcal{C}^1$, and $
f_e(x,y)=\partial \tilde{\ell}_e^{\theta_2}(y,x)/\partial y-\partial \tilde{\ell}_e^{\theta_1}(y,x)/\partial x$ for all $e\in\E$. 
\end{corollary}

\begin{proof} See Appendix~\ref{appendix:C}.
\end{proof}

Throughout this subsection, we assumed that all the drivers portray similar sensitivity to the imposed tolls. This is indeed a source of conservatism, specially when dealing with routing games in which the heterogeneity is caused by the fact that the drivers react differently to the imposed tolls. Certainly, an avenue for future research is to develop tolls for a more general setup.

\subsection{Example: Routing Game with Platooning Incentives}  \label{subsec:Platoon:tax}
Let us examine the possibility of finding a set of tolls that satisfies the conditions of Propositions~\ref{prop:toll1} and~\ref{prop:toll2} for the heterogeneous routing game introduced in Subsection~\ref{subsec:Platoon}. 
\begin{itemize}
\item \textbf{Distinguishable Types-Case 1:}  Substituting (\ref{eqn:derivative:1}) and (\ref{eqn:derivative:2}) into the condition of Proposition~\ref{prop:toll1} results in
\begin{equation} \label{eqn:distinguishable:example}
\begin{split}
&\frac{\partial \tilde{\tau}_e^{\mathrm{c}}(\phi_e^{\mathrm{c}},\phi_e^{\mathrm{t}})}{\partial \phi_e^{\mathrm{t}}}-\frac{\partial \tilde{\tau}_e^{\mathrm{t}}(\phi_e^{\mathrm{c}},\phi_e^{\mathrm{t}})}{\partial \phi_e^{\mathrm{c}}}=2c_0\alpha b_ea_e.
\end{split}
\end{equation}
Following simple algebraic calculations, we can check that the tolls $\tilde{\tau}_e^{\mathrm{c}}(\phi_e^{\mathrm{c}},\phi_e^{\mathrm{t}})=0$ and $\tilde{\tau}_e^{\mathrm{t}}(\phi_e^{\mathrm{c}},\phi_e^{\mathrm{t}})=(2c_0\alpha b_ea_e)\phi_e^{\mathrm{c}}$ satisfy~(\ref{eqn:distinguishable:example}). 
Noticing that $\tilde{\tau}_e^{\mathrm{t}}(\phi_e^{\mathrm{c}},\phi_e^{\mathrm{t}})\leq 0$ because by definition $a_e\in\mathbb{R}_{\leq 0}$, these terms can be interpreted as subsidies paid to the trucks to compensate for the fuel that is wasted due to presence of the cars on that specific edge.
\item \textbf{Distinguishable Types-Case 2:}  Another example of appropriate tolls is $\tilde{\tau}_e^{\mathrm{t}}(\phi_e^{\mathrm{c}},\phi_e^{\mathrm{t}})=0$ and $\tilde{\tau}_e^{\mathrm{c}}(\phi_e^{\mathrm{c}},\phi_e^{\mathrm{t}})=(-2c_0\alpha b_ea_e)\phi_e^{\mathrm{t}}$. Now, we have $\tilde{\tau}_e^{\mathrm{c}}(\phi_e^{\mathrm{c}},\phi_e^{\mathrm{t}})\geq 0$. In this case, the cars pay directly for the increased fuel consumption of the trucks and, therefore, they are inclined to travel along the edges that trucks do not use.
\item \textbf{Indistinguishable Types:}  For this case, using Corollary~\ref{cor:toll3}, it is easy to see that tolls $\tilde{\tau}_e(\phi_e^{\mathrm{c}},\phi_e^{\mathrm{t}})=(2c_0\alpha b_ea_e)\phi_e^{\mathrm{t}}$ work fine. We use these tolls in the numerical example developed in Section~\ref{sec:NumericalExample}. 
\end{itemize}

\section{Price of Anarchy for Affine Cost Functions}  \label{sec:inefficiency}
In the routing game literature, it is a widely known fact that generally, a Nash equilibrium is inefficient even when dealing with homogeneous routing games; see~\cite{roughgarden2007routing,roughgarden2002bad,6426526}. To quantify this inefficiency, many studies have used \textit{Price of Anarchy} (PoA) as a metric.

\subsection{Social Cost Function}
First, let us define the social cost of a flow vector $f=(f_{p'}^{\theta'})_{p'\in \P,\theta'\in\Theta}$ 
as
\begin{equation*}
\begin{split}
C(f)&\triangleq \sum_{p\in\P}\sum_{\theta\in\Theta} f^\theta_p \ell_p^\theta(f) \\&=\sum_{e\in\E} \sum_{\theta\in\Theta} \phi_e^\theta \tilde{\ell}_e^\theta((\phi_e^{\theta'})_{\theta'\in\Theta}),
\end{split}
\end{equation*}
where the second equality can be easily obtained by rearranging the terms. Using this social cost, we can define the optimal flow that we will use later for comparison with the Nash equilibrium.

\begin{definition}\textsc{(Socially Optimal Flow)} \label{def:socialoptimal}  $f\in\mathcal{F}$ is a socially optimal flow if $C(f)\leq C(\bar{f})$ for all $\bar{f}\in\mathcal{F}$.
\end{definition}

\begin{definition}\textsc{(PoA)} \label{def:PoA} The price of anarchy is defined as
$$
\mathrm{PoA}=\sup_{f^{\mathrm{Nash}}\in\mathcal{N}} \frac{C(f^{\mathrm{Nash}})}{\min_{f\in\mathcal{F}} C(f)},
$$
where $\mathcal{N}$ denotes the set of Nash equilibria of the heterogeneous routing game. In this definition, we follow the convention that ``$\frac{0}{0}=1$''.
\end{definition}

\subsection{Bounding the Price of Anarchy for Two Types with Affine Cost Functions}
Here, we present an upper bound for the inefficiency of the Nash equilibrium in heterogeneous routing games when $|\Theta|=2$. The edge cost functions are taken to be affine functions of the form
\begin{equation*}
\begin{split}
\ell_{e}^{\theta_1}(\phi_{e}^{\theta_1},\phi_{e}^{\theta_2})&=\alpha_{\theta_1\theta_1}^{e}\phi_{e}^{\theta_1}+\alpha_{\theta_1\theta_2}^{e}\phi_{e}^{\theta_2}+\beta_{\theta_1}^{e},\\
\ell_{e}^{\theta_2}(\phi_{e}^{\theta_1},\phi_{e}^{\theta_2})&=\alpha_{\theta_2\theta_1}^{e}\phi_{e}^{\theta_1}+\alpha_{\theta_2\theta_2}^{e}\phi_{e}^{\theta_2}+\beta_{\theta_2}^{e},
\end{split}
\end{equation*}
where $\alpha_{\theta_1\theta_1}^{e},\alpha_{\theta_1\theta_2}^{e},\alpha_{\theta_2\theta_1}^{e},\alpha_{\theta_2\theta_2}^{e},\beta_{\theta_1}^{e}, \beta_{\theta_2}^{e}\in\mathbb{R}_{\geq 0}$ are parameters of the routing game for each edge $e\in\E$. Notice that the condition $\alpha_{\theta_1\theta_1}^{e},\alpha_{\theta_1\theta_2}^{e},\alpha_{\theta_2\theta_1}^{e},\alpha_{\theta_2\theta_2}^{e}\in\mathbb{R}_{\geq 0}$ implies that the cost of using an edge is increasing in each flow separately (i.e., when a driver of any type switches to an edge, she cannot decrease the cost of the users on this new edge) while $\beta_{\theta_1}^{e}, \beta_{\theta_2}^{e}\in\mathbb{R}_{\geq 0}$ implies that the starting cost of using a road is non-negative.  This assumption is certainly stronger than Assumption~\ref{assum:1}. Subsection~\ref{subsec:Platoon} presents a motivating example for affine cost functions.

\begin{theorem} \label{tho:PoA} Let 
\begin{subequations}
\begin{equation} \label{eqn:condition:PoA:1}
\alpha_{\theta_2\theta_1}^{e}=\alpha_{\theta_1\theta_2}^{e}
\end{equation}
\begin{equation} \label{eqn:condition:PoA}
\matrix{cc}{\alpha_{\theta_1\theta_1}^{e} & \alpha_{\theta_1\theta_2}^{e} \\ \alpha_{\theta_2\theta_1}^{e} & \alpha_{\theta_2\theta_2}^{e}}\geq 0,
\end{equation}
\end{subequations}
for all $ e\in\E$. Then, $\mathrm{PoA}\leq2 $.
\end{theorem}

\begin{proof} First, note that if $\alpha_{\theta_2\theta_1}^{e}=\alpha_{\theta_1\theta_2}^{e}$ for all $e\in\E$, the condition of Corollary~\ref{cor:optimization} is satisfied. Therefore, we can easily calculate the potential function as 
\begin{small}
\begin{equation} \label{eqn:derivationofV}
\begin{split}
V(f)=&\sum_{e\in\E}\bigg[\frac{1}{2}\alpha_{\theta_1\theta_1}^{e}(\phi_{e}^{\theta_1})^2\hspace{-.03in}+\hspace{-.03in}(\alpha_{\theta_1\theta_2}^{e}\phi_{e}^{\theta_2}\hspace{-.03in}+\hspace{-.03in}\beta_{\theta_1}^{e})\phi_{e}^{\theta_1}\hspace{-.03in}+\hspace{-.03in}\frac{1}{2}\alpha_{\theta_2\theta_2}^{e}(\phi_{e}^{\theta_2})^2+(\alpha_{\theta_2\theta_1}^{e}\phi_{e}^{\theta_1}+\beta_{\theta_2}^{e})\phi_{e}^{\theta_2} \hspace{-.03in}-\hspace{-.03in}\alpha_{\theta_1\theta_2}^{e}\phi_{e}^{\theta_1}\phi_{e}^{\theta_2} \bigg]\\
=&\sum_{e\in\E}\bigg[\frac{1}{2}\phi_{e}^{\theta_1}\ell_{e}^{\theta_1}(\phi_{e}^{\theta_1},\phi_{e}^{\theta_2})+\frac{1}{2}\beta_{\theta_1}^{e}\phi_{e}^{\theta_1} +\frac{1}{2}\phi_{e}^{\theta_2}\ell_{e}^{\theta_2}(\phi_{e}^{\theta_1},\phi_{e}^{\theta_2})+\frac{1}{2}\beta_{\theta_2}^{e}\phi_{e}^{\theta_2}  \bigg]\\
=&\frac{1}{2}C(f)+\sum_{e\in\E}\frac{1}{2}\bigg[\beta_{\theta_1}^{e}\phi_{e}^{\theta_1}+\beta_{\theta_2}^{e}\phi_{e}^{\theta_2}  \bigg],
\end{split}
\end{equation}
\end{small}
Furthermore, following the argument of~\cite[p.\,71]{boyd2004convex}, we know that the social cost function is a convex function if and only if~(\ref{eqn:condition:PoA}) is satisfied. Notice that~\eqref{eqn:derivationofV} shows that the potential function $V$ is a convex function if the social cost function $C$ is a convex function (because the summation of a convex function and a linear function is a convex function). Let us use $\bar{f}$ and $f$ to denote the Nash equilibrium and the socially optimal flow, respectively. Now, we can prove inequality
\begin{small}
\begin{align}
C(\bar{f})& \leq 2 V(\bar{f}) \nonumber && \mbox{by}~\eqref{eqn:derivationofV}\mbox{ and $\beta_{\theta_1}^e,\beta_{\theta_2}^e\in\mathbb{R}_{\geq 0}$} \\[-.5em] &\leq 2V(f) &&  \nonumber \mbox{by Corollary~\ref{cor:convexprogramming}} \\[-.5em]
&\leq 2\sum_{e\in\E}\bigg(\int_0^{\phi_e^{\theta_1}} \hspace{-.07in}\tilde{\ell}_e^{\theta_1}(u_1,\phi_e^{\theta_2}) \d u_1+\int_0^{\phi_e^{\theta_2}}\hspace{-.07in} \tilde{\ell}_e^{\theta_2}(\phi_e^{\theta_1},u_2) \d u_2 \nonumber \\[-.5em] &\hspace{.3in}-\int_0^{\phi_e^{\theta_2}}\hspace{-.08in}\int_0^{\phi_e^{\theta_1}} \hspace{-.07in}\frac{\partial}{\partial u } \tilde{\ell}_e^{\theta_1} (t,u) \d t \d u \bigg) && \mbox{by Definition of $V$}  \nonumber
\\[-.5em] &\leq 2\sum_{e\in\E}\left(\int_0^{\phi_e^{\theta_1}} \hspace{-.07in}\tilde{\ell}_e^{\theta_1}(u_1,\phi_e^{\theta_2}) \d u_1+\int_0^{\phi_e^{\theta_2}}\hspace{-.07in} \tilde{\ell}_e^{\theta_2}(\phi_e^{\theta_1},u_2) \d u_2 \right) && \mbox{by } \alpha_{\theta_1\theta_2}^{e},\alpha_{\theta_2\theta_1}^{e}\in\mathbb{R}_{\geq 0} \nonumber
\\[-.5em] &\leq 2 \bigg(\sum_{e\in\E}\int_0^{\phi_e^{\theta_1}} \left[\tilde{\ell}_e^{\theta_1}(u_1,\phi_e^{\theta_2})+u_1\frac{\partial}{\partial u_1}\tilde{\ell}_e^{\theta_1}(u_1,\phi_e^{\theta_2}) \right] \d u_1 && \nonumber \\[-.5em] & \hspace{0.3in}+ \sum_{e\in\E} \int_0^{\phi_e^{\theta_2}} \left[\tilde{\ell}_e^{\theta_2}(\phi_e^{\theta_1},u_2)+u_2 \frac{\partial}{\partial u_2}\tilde{\ell}_e^{\theta_2}(\phi_e^{\theta_1},u_2) \right] \d u_2 \bigg) && \mbox{by }\alpha_{\theta_1\theta_1}^{e},\alpha_{\theta_2\theta_2}^{e}\in\mathbb{R}_{\geq 0} \nonumber\\[-.5em] &\leq 2 \left( \sum_{e\in\E}\phi_e^{\theta_1}\tilde{\ell}_e^{\theta_1}(\phi_e^{\theta_1},\phi_e^{\theta_2})+ \sum_{e\in\E}\phi_e^{\theta_2}\tilde{\ell}_e^{\theta_2}(\phi_e^{\theta_1},\phi_e^{\theta_2}) \right) & \nonumber \\[-.5em] & =
 2 C(f). \label{longequation:5} 
\end{align}
\end{small}
This completes the proof.
\end{proof}

Notice that in many practical situations (such as the one presented in~Subsection~\ref{subsec:Platoon:potential} for routing games with platooning incentives), $\alpha_{\theta_2\theta_1}^{e}\neq\alpha_{\theta_1\theta_2}^{e}$. Therefore, we may not be able to use Theorem~\ref{tho:PoA} to find an upper bound for the PoA. However, as also discussed in Section~\ref{sec:tolls}, in some cases, we  might be able to manipulate these gains through appropriate tolls to make sure~\eqref{eqn:condition:PoA:1} holds. In addition, condition~\eqref{eqn:condition:PoA} is equivalent to the condition that $\alpha_{\theta_1\theta_1}^{e}\alpha_{\theta_2\theta_2}^{e}\geq \alpha_{\theta_2\theta_1}^{e}\alpha_{\theta_1\theta_2}^{e}$ for all $e\in\E$. This condition intuitively means that cost function of each type of vehicles is more influenced by the flow of its own type than the flow of the other type. This condition may not hold in general in transportation networks. In such case, instead of using Corollary~\ref{cor:convexprogramming}, we may use Corollary~\ref{cor:optimization} in the proof of Theorem~\ref{tho:PoA} (that is the only place that we use the convexity of the potential function which we proved using the convexity of the social decision function). However, doing so, we cannot bound the ratio $C(f^{\mathrm{Nash}})/\min_{f} C(f)$ for all $f^{\mathrm{Nash}}\in\mathcal{N}$. Therefore, instead of showing that PoA is bounded from above by two, we can then only show that the Price of Stability\footnote{ Price of Stability (PoS), or commonly known as the optimistic Price of Anarchy, is defined as $\inf_{f^{\mathrm{Nash}}\in\mathcal{N}} C(f^{\mathrm{Nash}})/\min_{f\in\mathcal{F}} C(f)$; note that we use $\inf$ operator instead of $\sup$ operator in this definition in contrast to that of Definition~\ref{def:PoA}. See~\cite{anshelevich2003near} for more explanation regarding the difference between PoS and PoA. } is upper bounded by two (because we can show that the ratio is bounded by two for only one Nash equilibrium and not for all Nash equilibria).

\begin{table*}[t]
\centering
\footnotesize
\caption{ \label{table:1} Parameters of the heterogeneous routing game in the numerical example.}
\begin{tabular}{|c|c|c|c|c|c|c|c|c|c|c|c|c|} \hline
 & $e_0$ & $e_1$ & $e_2$ & $e_3$ & $e_4$ & $e_5$ & $e_6$ & $e_7$ & $e_8$ & $e_9$ &  $e_{10}$ &$e_{11}$ 
\\ \hline
$\alpha_{\mathrm{a}\mathrm{a}}$ &
1.0 & 2.0 & 3.0 & 1.0 & 4.0 & 0.5 & 1.0 & 1.0 & 2.0 & 1.0 & 4.0 & 1.0
\\ \hline
$\alpha_{\mathrm{a}\mathrm{t}}$ &
0.6 & 0.4 & 0.1 & 0.1 & 0.5 & 0.1 & 0.7 & 0.1 & 0.1 & 0.2 & 0.1 & 0.3
\\ \hline
$\alpha_{\mathrm{t}\mathrm{t}}$ &
2.0 & 3.0 & 1.0 & 0.8 & 1.0 & 1.0 & 1.5 & 3.0 & 1.7 & 3.0 & 1.0 & 1.3
\\ \hline
$\beta_{\mathrm{a}}$ &
2.0 & 2.0 & 4.5 & 2.0 & 2.0 & 4.5 & 2.0 & 2.0 & 4.5 & 2.0 & 2.0 & 4.5
\\ \hline
$\beta_{\mathrm{t}}$ & 
4.0 & 4.0 & 1.5 & 4.0 & 4.0 & 1.5 & 4.0 & 4.0 & 1.5 & 4.0 & 4.0 & 1.5
\\ \hline
\end{tabular}
\end{table*}

\section{Numerical Example} \label{sec:NumericalExample}
In this section, we present a numerical example motivated by the routing game with platooning incentives in Subsection~\ref{subsec:Platoon}. We use the graph $\G=(\V,\E)$ in Figure~\ref{figure:1}. We have three commodities $(s_1,t_1)=(0,1)$, $(s_2,t_2)=(2,3)$, and $(s_3,t_3)=(7,8)$. The corresponding paths for the commodities are
\begin{align*}
&\P_1\hspace{-0.03in}=\hspace{-0.03in}\{\hspace{-0.03in}\{e_1\},\hspace{-0.03in}\{e_2,e_4,e_3\},\hspace{-0.03in}\{e_2,e_7,e_5\}\hspace{-0.03in}\},\\
&\P_2\hspace{-0.03in}=\hspace{-0.03in}\{\hspace{-0.03in}\{e_{10}\},\hspace{-0.03in}\{e_9,e_7,e_8\},\hspace{-0.03in}\{e_9,e_4,e_6\}\hspace{-0.03in}\},\\
&\P_3\hspace{-0.03in}=\hspace{-0.03in}\{\hspace{-0.03in}\{e_{11},e_{10},e_0\},\hspace{-0.03in}\{e_{11},e_9,e_7,e_8,e_0\},\hspace{-0.03in}\{e_{11},e_9,e_4,e_6,e_0\}\hspace{-0.03in}\}.
\end{align*}
The edge cost functions are taken to be affine functions of the form
\begin{equation*}
\begin{split}
\tilde{\ell}_{e_i}^{\mathrm{c}}(\phi_{e_i}^{\mathrm{c}},\phi_{e_i}^{\mathrm{t}})&=\alpha_{\mathrm{c}\mathrm{c}}^{(i)}\phi_{e_i}^{\mathrm{c}}+\bar{\alpha}_{\mathrm{c}\mathrm{t}}^{(i)}\phi_{e_i}^{\mathrm{t}}+\beta_{\mathrm{c}}^{(i)},\\
\tilde{\ell}_{e_i}^{\mathrm{t}}(\phi_{e_i}^{\mathrm{c}},\phi_{e_i}^{\mathrm{t}})&=\alpha_{\mathrm{t}\mathrm{c}}^{(i)}\phi_{e_i}^{\mathrm{c}}+\bar{\alpha}_{\mathrm{t}\mathrm{t}}^{(i)}\phi_{e_i}^{\mathrm{t}}+\beta_{\mathrm{t}}^{(i)},
\end{split}
\end{equation*}
where the definitions and the physical intuition of the parameters $\alpha_{\mathrm{c}\mathrm{c}}^{(i)},\alpha_{\mathrm{t}\mathrm{c}}^{(i)},\bar{\alpha}_{\mathrm{t}\mathrm{t}}^{(i)},\bar{\alpha}_{\mathrm{c}\mathrm{t}}^{(i)},\beta_{\mathrm{c}}^{(i)},\beta_{\mathrm{t}}^{(i)}$ can be found in Subsection~\ref{subsec:Platoon}. Recalling that $\bar{\alpha}_{\mathrm{t}\mathrm{c}}^{(i)}\neq\bar{\alpha}_{\mathrm{c}\mathrm{t}}^{(i)}$ (see Subsection~\ref{subsec:Platoon:potential}), the condition of Corollary~\ref{cor:optimization} is not satisfied. Therefore, we use the tax $\tilde{\tau}_{e_i}(\phi_{e_i}^{\mathrm{c}},\phi_{e_i}^{\mathrm{t}})=(2c_0\alpha b_ea_e)\phi_{e_i}^{\mathrm{t}}$ which is developed in Subsection~\ref{subsec:Platoon:tax}. This results in
\begin{equation*}
\begin{split}
\tilde{\ell}_{e_i}^{\mathrm{c}}(\phi_{e_i}^{\mathrm{c}},\phi_{e_i}^{\mathrm{t}})+\tilde{\tau}_{e_i}(\phi_{e_i}^{\mathrm{c}},\phi_{e_i}^{\mathrm{t}})&=\alpha_{\mathrm{c}\mathrm{c}}^{(i)}\phi_{e_i}^{\mathrm{c}}+\alpha_{\mathrm{c}\mathrm{t}}^{(i)}\phi_{e_i}^{\mathrm{t}}+\beta_{\mathrm{c}}^{(i)},\\
\tilde{\ell}_{e_i}^{\mathrm{t}}(\phi_{e_i}^{\mathrm{c}},\phi_{e_i}^{\mathrm{t}})+\tilde{\tau}_{e_i}(\phi_{e_i}^{\mathrm{c}},\phi_{e_i}^{\mathrm{t}})&=\alpha_{\mathrm{t}\mathrm{c}}^{(i)}\phi_{e_i}^{\mathrm{c}}+\alpha_{\mathrm{t}\mathrm{t}}^{(i)}\phi_{e_i}^{\mathrm{t}}+\beta_{\mathrm{t}}^{(i)},
\end{split}
\end{equation*}
where $\alpha_{\mathrm{c}\mathrm{t}}^{(i)}=\bar{\alpha}_{\mathrm{c}\mathrm{t}}^{(i)}+2c_0\alpha b_ea_e$ and $\alpha_{\mathrm{t}\mathrm{t}}^{(i)}=\bar{\alpha}_{\mathrm{t}\mathrm{t}}^{(i)}+2c_0\alpha b_ea_e$.
In this case, we can calculate the potential function as
\begin{equation*}
\begin{split}
V&\hspace{-.03in}=\hspace{-.04in}\sum_{i=0}^{11} \hspace{-.04in}\bigg[\frac{1}{2}\alpha_{\mathrm{c}\mathrm{c}}^{(i)}(\phi_{e_i}^{\mathrm{c}})^2\hspace{-.03in}+\hspace{-.03in}(\alpha_{\mathrm{c}\mathrm{t}}^{(i)}\phi_{e_i}^{\mathrm{t}}\hspace{-.03in}+\hspace{-.03in}\beta_{\mathrm{c}}^{(i)})\phi_{e_i}^{\mathrm{c}}-\alpha_{\mathrm{c}\mathrm{t}}^{(i)}\phi_{e_i}^{\mathrm{c}}\phi_{e_i}^{\mathrm{t}}+\frac{1}{2}\alpha_{\mathrm{t}\mathrm{t}}^{(i)}(\phi_{e_i}^{\mathrm{t}})^2+(\alpha_{\mathrm{t}\mathrm{c}}^{(i)}\phi_{e_i}^{\mathrm{c}}+\beta_{\mathrm{t}}^{(i)})\phi_{e_i}^{\mathrm{t}}  \bigg].
\end{split}
\end{equation*}
Noticing that solving a non-convex quadratic programming problem might be numerically intractable in general, we focus on the case in which the potential function is a convex function. Following the argument of~\cite[p.\,71]{boyd2004convex}, we know that the potential function is a convex function if and only if
$$
\matrix{cc}{\alpha_{\mathrm{c}\mathrm{c}}^{(i)} & \frac{1}{2}\alpha_{\mathrm{c}\mathrm{t}}^{(i)} \\ \frac{1}{2}\alpha_{\mathrm{t}\mathrm{c}}^{(i)} & \alpha_{\mathrm{t}\mathrm{t}}^{(i)}}\geq 0, \;\; \forall i=\{0,\dots,11\}.
$$
Let us pick the parameters for the routing game according to Table~\ref{table:1}. Furthermore, we choose $(F_1^{\mathrm{a}},F_1^{\mathrm{b}})=(5,1)$, $(F_2^{\mathrm{a}},F_2^{\mathrm{b}})=(3,3)$, and $(F_3^{\mathrm{a}},F_3^{\mathrm{b}})=(2,4)$. After solving the optimization problem in Corollary~\ref{cor:optimization}, we can extract the path flows and path cost functions shown in Table~\ref{table:2} which demonstrate a Nash equilibrium (see Definition~\ref{def:Nashroutingtax})\footnote{See~\texttt{http://dl.dropbox.com/u/36867745/HeterogeneousRoutingGame.zip} for the Python code to simulate this numerical example. }. 
\begin{figure}
\centering
\begin{tikzpicture}[scale=1.2,>=stealth']
\tikzset{mystyle1/.style={circle,draw,line width=1pt,minimum size=0.8cm,scale=0.7}}
\tikzset{mystyle2/.style={circle,draw,line width=1pt,minimum size=0.8cm,scale=0.7}}
\tikzset{mystyle3/.style={circle,draw,line width=1pt,minimum size=0.8cm,scale=0.7}}
\node (v2) at (+0.0,+0.0) [mystyle1] {$2$};
\node (v3) at (+3.0,+0.0) [mystyle2] {$3$};
\node (v0) at (+0.0,+2.0) [mystyle1] {$0$};
\node (v1) at (+3.0,+2.0) [mystyle2] {$1$};
\node (v4) at (+1.0,+1.0) [mystyle3] {$4$};
\node (v5) at (+2.0,+1.3) [mystyle3] {$5$};
\node (v6) at (+2.0,+0.7) [mystyle3] {$6$};
\node (v7) at (-1.2,+0.0) [mystyle1] {$7$};
\node (v8) at (+4.2,+0.0) [mystyle2] {$8$};
\draw[->,line width=1pt] (v2) to node[below]{\scriptsize $e_{10}$} (v3);
\draw[->,line width=1pt] (v0) to node[above]{\scriptsize $e_1$   } (v1);
\draw[->,line width=1pt] (v2) to node[below]{\scriptsize $e_9$   } (v4);
\draw[->,line width=1pt] (v0) to node[above]{\scriptsize $e_2$   } (v4);
\draw[->,line width=1pt] (v4) to node[above]{\scriptsize $e_4$   } (v5);
\draw[->,line width=1pt] (v4) to node[below]{\scriptsize $e_7$   } (v6);
\draw[->,line width=1pt] (v5) to node[above]{\scriptsize $e_6$   } (v3);
\draw[->,line width=1pt] (v5) to node[above]{\scriptsize $e_3$   } (v1);
\draw[->,line width=1pt] (v6) to node[below]{\scriptsize $e_8$   } (v3);
\draw[->,line width=1pt] (v6) to node[below]{\scriptsize $e_5$   } (v1);
\draw[->,line width=1pt] (v7) to node[below]{\scriptsize $e_{11}$} (v2);
\draw[->,line width=1pt] (v3) to node[below]{\scriptsize $e_0$   } (v8);
\end{tikzpicture}
\caption{\label{figure:1} Transportation network in the numerical example.}
\end{figure}
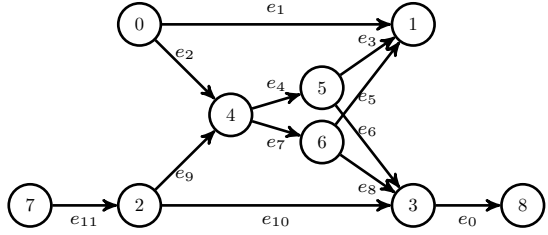
In addition, we can calculate
\begin{align*}
\frac{C(f)}{C(f^*)}&=1.0137
\leq 2 = \mathrm{Upper\;Bound\;of\;the\;PoA},
\end{align*}
where $f^*$ denotes the socially optimal flow. This shows that the social cost of the recovered Nash equilibrium is only $1.0137$ times the cost of the socially optimal solutions.

\begin{table}
\begin{center}
\caption{\label{table:2} The path flow and path cost function at a Nash equilibrium extracted by minimizing the potential function. }
\footnotesize
$$
\begin{array}{cc}
\begin{tabular}{|c|c|c|} \hline
& $f^{\mathrm{c}}_p$ & $f^{\mathrm{t}}_p$ \\ \hline
\multirow{3}{*}{$p\in\P_{1}$} & 4.97 & 0.79 \\ \cline{2-3}
& 0.00 & 0.00 \\ \cline{2-3}
& 0.03 & 0.21 \\ \hline
\multirow{3}{*}{$p\in\P_{2}$} & 1.04 & 0.06 \\ \cline{2-3}
& 0.04 & 0.00 \\ \cline{2-3}
& 1.92 & 2.94 \\ \hline
\multirow{3}{*}{$p\in\P_{3}$} & 0.01 & 0.00 \\ \cline{2-3}
& 1.10 & 0.00 \\ \cline{2-3}
& 0.88 & 4.00 \\ \hline
\end{tabular}
&
\hspace{.2in}
\begin{tabular}{|c|c|c|} \hline
& $\ell_p^{\mathrm{c}}(f)$ & $\ell^{\mathrm{t}}_p(f)$ \\ \hline
\multirow{3}{*}{$p\in\P_{1}$} & 12.26 & 8.35 \\ \cline{2-3}
& 13.18 & 10.29 \\ \cline{2-3}
& 12.26 & 8.35 \\ \hline
\multirow{3}{*}{$p\in\P_{2}$} & 13.92 & 11.22 \\ \cline{2-3}
& 13.92 & 13.98 \\ \cline{2-3}
& 13.92 & 11.22 \\ \hline
\multirow{3}{*}{$p\in\P_{3}$} & 28.02 & 31.73 \\ \cline{2-3}
& 28.02 & 34.48 \\ \cline{2-3}
& 28.02 & 31.72 \\ \hline
\end{tabular}
\end{array}
$$
\end{center}
\end{table}

\section{Conclusions} \label{sec:conclusions}
In this article, we proposed a heterogeneous routing game in which the players may belong to more than one type. The type of each player determines the cost of using an edge as a function of the flow of all types over that edge. We proved that this heterogeneous routing game admits at least one Nash equilibrium. Additionally, we gave a necessary and sufficient condition for the existence of a potential function, which indeed implies that we can transform the problem of finding a Nash equilibrium into an optimization problem. Finally, we developed tolls to guarantee the existence of a potential function. Possible future research will focus on generalizing these results to higher number of types or a continuum of player types. 

\bibliographystyle{ieeetr}
\bibliography{ref}

\appendix

\section{Proof of Corollary~\ref{cor:convexprogramming} } \label{appendix:A}
Let us define the Lagrangian as
\begin{equation*}
\begin{split}
L((f_{p'}^{\theta_1})_{p'\in\P},(f_{p'}^{\theta_2})_{p'\in\P})=&V((f_{p'}^{\theta_1})_{p'\in\P},(f_{p'}^{\theta_2})_{p'\in\P})+\sum_{i=1}^2\sum_{e\in E} v_e^i \left(\sum_{p\in\P:e\in p}f_p^{\theta_i}-\phi_e^{\theta_i}\right)\\ &-\sum_{i=1}^2\sum_{k=1}^K w_k^i \left(\sum_{p\in\P_k} f_p^{\theta_i}-F_k^{\theta_i} \right)-\sum_{i=1}^2\sum_{p\in\P} \lambda_p^i f_p^i,
\end{split}
\end{equation*}
where $(v_e^1)_{e\in\E}\in\mathbb{R}^{|\mathcal{E}|}$, $(v_e^2)_{e\in\E}\in\mathbb{R}^{|\mathcal{E}|}$, $(w_k^1)_{k\in\K}\in\mathbb{R}^{K}$, $(w_k^2)_{k\in\K}\in\mathbb{R}^{K}$, $(\lambda_p^1)_{p\in\P}\in\mathbb{R}_{\geq 0}^{|\P|}$, and $(\lambda_p^2)_{p\in\P}\in\mathbb{R}_{\geq 0}^{|\P|}$ are Lagrange multipliers. Using Karush--Kuhn--Tucker conditions~\cite[p.\,244]{boyd2004convex}, optimality conditions are
\begin{subequations} \label{eqn:partial:phi}
\begin{equation} \label{eqn:partial:phi:1}
\begin{split}
\frac{\partial L}{\partial \phi_e^{\theta_1}} &=\tilde{\ell}_e^{\theta_1}(\phi_e^{\theta_1},\phi_e^{\theta_2})+\int_0^{\phi_e^{\theta_2}} \hspace{-.07in} \frac{\partial \tilde{\ell}_e^{\theta_2}(\phi_e^{\theta_1},u_2)}{\partial \phi_e^{\theta_1}}  \d u_2 -\hspace{-.04in} \int_0^{\phi_e^{\theta_2}} \hspace{-.07in} \frac{\partial}{\partial u } \tilde{\ell}_e^{\theta_1} (\phi_e^{\theta_1},u) \d u - v_e^1\\&=\tilde{\ell}_e^{\theta_1}(\phi_e^{\theta_1},\phi_e^{\theta_2})- v_e^1+\hspace{-.04in}\int_0^{\phi_e^{\theta_2}} \hspace{-.09in} \bigg(\hspace{-.02in}\frac{\partial \tilde{\ell}_e^{\theta_2}(\phi_e^{\theta_1},u)}{\partial \phi_e^{\theta_1}}  \hspace{-0.04in}-\hspace{-0.04in} \frac{\partial \tilde{\ell}_e^{\theta_1} (\phi_e^{\theta_1},u)}{\partial u }  \hspace{-.03in}\bigg) \d u 
\\&=\tilde{\ell}^{\theta_1}_e(\phi_e^{\theta_1},\phi_e^{\theta_2})-v_e^1=0, \;\; \forall e\in\E,
\end{split}
\end{equation}
\begin{equation} \label{eqn:partial:phi:2}
\begin{split}
\frac{\partial L}{\partial \phi_e^{\theta_2}} &=\int_0^{\phi_e^{\theta_1}} \hspace{-.07in} \frac{\partial \tilde{\ell}_e^{\theta_1}(u_1,\phi_e^{\theta_2})}{\partial \phi_e^{\theta_2}}  \d u_1 + \tilde{\ell}_e^{\theta_2}(\phi_e^{\theta_1},\phi_e^{\theta_2}) - \int_0^{\phi_e^{\theta_1}} \hspace{-.07in} \frac{\partial}{\partial \phi_e^{\theta_2}} \tilde{\ell}_e^{\theta_1} (t,\phi_e^{\theta_2}) \d t - v_e^2
\\&=\tilde{\ell}_e^{\theta_2}(\phi_e^{\theta_1},\phi_e^{\theta_2})-v_e^2=0, \;\; \forall e\in\E,
\end{split}
\end{equation}
\end{subequations}
and
\begin{subequations} \label{eqn:partial:f}
\begin{equation} \label{eqn:partial:f:1}
\begin{split}
&\frac{\partial }{\partial f_p^{\theta_1}} L=\sum_{e\in p} v_e^1-w_k^1-\lambda_p^1=0, \;\; \forall p\in\P,
\vspace{-.2in}
\end{split}
\end{equation}
\begin{equation} \label{eqn:partial:f:2}
\begin{split}
&\frac{\partial }{\partial f_p^{\theta_2}} L=\sum_{e\in p} v_e^2-w_k^2-\lambda_p^2=0, \;\; \forall p\in\P.
\end{split}
\end{equation}
\end{subequations}
In addition, the complimentary slackness conditions for inequality constraints result in $\lambda_p^1f_p^1=0$ and $\lambda_p^2f_p^2=0$ for all $p\in\P$. Hence, for all $k$ and $p\in\P_k$, we have
\begin{align}
\ell_p^{\theta_i}(f)&=\sum_{e\in p} \tilde{\ell}^{\theta_i}_e(\phi_e^{\theta_1},\phi_e^{\theta_2}) & \nonumber\\[-.5em]
&=\sum_{e\in p} v_e^i && \mbox{by~\eqref{eqn:partial:phi}}\;\; \nonumber \\[-.5em] &= w_k^i+\lambda_p^i. && \mbox{by~\eqref{eqn:partial:f}} \nonumber
\end{align}
Thus, if $f_p^{\theta_i},f_{p'}^{\theta_i}>0$, using complimentary slackness, we  get $\lambda_p^{\theta_i}=0$ and $\lambda_{p'}^{\theta_i}=0$, which results in
$$
\ell_p^{\theta_i}(f)=\ell_{p'}^{\theta_i}(f)=w_k^i.
$$
Additionally, for all $p''\in\P_k$, where $f_{p''}^{\theta_i}=0$, we have $\lambda_p^{\theta_i}\geq 0$ (because of dual feasibility), which results in
$$
\ell_{p''}^{\theta_i}(f)=w_k^i+\lambda_{p''}^{\theta_i}\geq w_k^i=\ell_p^{\theta_i}(f).
$$
This is the definition of a Nash equilibrium.

\section{Proof of Proposition~\ref{prop:toll1} } \label{appendix:B}
Note that introducing the tolls $\tilde{\tau}_e^\theta(\phi_e^{\theta_1},\phi_e^{\theta_2})$ has the same impact on the routing game as replacing the edge cost functions in the original heterogeneous routing game from $\tilde{\ell}_e^\theta(\phi_e^{\theta_1},\phi_e^{\theta_2})$ to $\tilde{\ell}_e^\theta(\phi_e^{\theta_1},\phi_e^{\theta_2})+\tilde{\tau}_e^\theta(\phi_e^{\theta_1},\phi_e^{\theta_2})$. Thanks to Lemma~\ref{lemma:sufficiency}, the abstract game based upon this new heterogeneous routing game admits the potential function $V$ if 
\begin{equation*}
\begin{split}
&\frac{\partial (\tilde{\ell}_e^{\theta_1}(\phi_e^{\theta_1},\phi_e^{\theta_2})+\tilde{\tau}_e^{\theta_1}(\phi_e^{\theta_1},\phi_e^{\theta_2}))}{\partial \phi_e^{\theta_2}}-\frac{\partial (\tilde{\ell}_e^{\theta_2}(\phi_e^{\theta_1},\phi_e^{\theta_2})+\tilde{\tau}_e^{\theta_2}(\phi_e^{\theta_1},\phi_e^{\theta_2}))}{\partial \phi_e^{\theta_1}}=0.
\end{split}
\end{equation*}
With rearranging the terms in this equality, we can extract the condition in the statement of the proposition.

\section{Proof of Corollary~\ref{cor:toll3} } \label{appendix:C}
The proof can be seen as a direct application of the result of~\cite[Ch.\,4]{polianin2002handbook} to Proposition~\ref{prop:toll2}. However, let us show this fact following simple algebraic manipulations. Notice that
\begin{equation*}
\begin{split}
\frac{\partial \tilde{\tau}_e(\phi_e^{\theta_1},\phi_e^{\theta_2})}{\partial \phi_e^{\theta_1}}=&\frac{\partial}{\partial \phi_e^{\theta_1}}\bigg[c_e+\psi_e(\phi_e^{\theta_1}+\phi_e^{\theta_2})+\int_0^{\phi_e^{\theta_2}}\hspace{-.13in} f_e(q,\phi_e^{\theta_1}+\phi_e^{\theta_2}-q) \d q\bigg]\\=&
\frac{\d \psi_e(u)}{\d u}\big|_{u=\phi_e^{\theta_1}+\phi_e^{\theta_2}}+\int_0^{\phi_e^{\theta_2}}\hspace{-.05in} \frac{\partial f_e(q,u)}{\partial u}\big|_{u=\phi_e^{\theta_1}+\phi_e^{\theta_2}-q} \d q
\end{split}
\end{equation*}
and
\begin{equation*}
\begin{split}
\frac{\partial \tilde{\tau}_e(\phi_e^{\theta_1},\phi_e^{\theta_2})}{\partial \phi_e^{\theta_2}}=&\frac{\partial}{\partial \phi_e^{\theta_2}}\bigg[c_e+\psi_e(\phi_e^{\theta_1}+\phi_e^{\theta_2})+\int_0^{\phi_e^{\theta_2}}\hspace{-.13in} f_e(q,\phi_e^{\theta_1}+\phi_e^{\theta_2}-q) \d q\bigg]\\=&
\frac{\d \psi_e(u)}{\d u}\big|_{u=\phi_e^{\theta_1}+\phi_e^{\theta_2}}+
f_e(\phi_e^{\theta_2},\phi_e^{\theta_1})+\int_0^{\phi_e^{\theta_2}}\hspace{-.05in} \frac{\partial f_e(q,u)}{\partial u}\big|_{u=\phi_e^{\theta_1}+\phi_e^{\theta_2}-q} \d q.
\end{split}
\end{equation*}
Therefor, we get
\begin{equation*}
\begin{split}
\frac{\partial \tilde{\tau}_e(\phi_e^{\theta_1},\phi_e^{\theta_2})}{\partial \phi_e^{\theta_2}}-\frac{\partial \tilde{\tau}_e(\phi_e^{\theta_1},\phi_e^{\theta_2})}{\partial \phi_e^{\theta_1}}&=f_e(\phi_e^{\theta_2},\phi_e^{\theta_1})\\&=\frac{\partial \tilde{\ell}_e^{\theta_2}(\phi_e^{\theta_1},\phi_e^{\theta_2})}{\partial \phi_e^{\theta_1}}-\frac{\partial \tilde{\ell}_e^{\theta_1}(\phi_e^{\theta_1},\phi_e^{\theta_2})}{\partial \phi_e^{\theta_2}},
\end{split}
\end{equation*}
where the second equality directly follows from the definition of the mapping $f_e:\mathbb{R}_{\geq 0}\times\mathbb{R}_{\geq 0}\rightarrow\mathbb{R}$ in the statement of the corollary. Now, we can use Proposition~\ref{prop:toll2} to show that a potential function indeed exists.

\end{document}